\documentclass[aps,pra,twocolumn,superscriptaddress,floatfix,
nofootinbib,showpacs,longbibliography]{revtex4-2}
\usepackage[utf8]{inputenc}  
\usepackage[T1]{fontenc}     
\usepackage[british]{babel}  
\usepackage[sc,osf]{mathpazo}
\usepackage[scaled=0.86]{berasans}  
\usepackage[colorlinks=true, citecolor=blue, urlcolor=blue]{hyperref}  
\usepackage{graphicx}
\usepackage[babel]{microtype}  
\usepackage{amsmath,amssymb,amsthm,bm,amsfonts,mathrsfs,bbm} 
\usepackage{booktabs}
\usepackage{tabularx}
\usepackage{array}

\usepackage{xspace}  
\usepackage{comment}
\usepackage{pgf,tikz}
\usepackage{xcolor}
\usepackage{multirow}
\usepackage{array}
\usepackage{bigstrut}
\usepackage{braket}
\usepackage{color}
\usepackage{natbib}
\usepackage{multirow}
\usepackage{mathtools}
\usepackage{float}
\usepackage{xcolor,colortbl}
\usepackage{physics}
\usepackage{amsmath}
\usepackage{color}
\usepackage[justification=justified, format=plain]{subcaption}
\usepackage[justification=raggedright]{caption}

\newcommand{\be}{\begin{equation}}
\newcommand{\ee}{\end{equation}}
\newcommand{\ba}{\begin{eqnarray}}
\newcommand{\ea}{\end{eqnarray}}

\newtheorem{theorem}{Theorem}

\newtheorem{definition}{Definition}
\newtheorem{proposition}{Proposition}
\newtheorem{observation}{Observation}
\newtheorem{example}{Example}
\newtheorem{remark}{Remark}

\def\>{\rangle}
\def\<{\langle}

\usepackage{tabularx}
\usepackage{booktabs}

\begin{document}

\title{
Witness-based quantification of non-Markovianity  without complete Choi state reconstruction}	

\author{Pritam Roy}
\email{roy.pritamphy@gmail.com}
\affiliation{S. N. Bose National Centre for Basic Sciences, Block JD, Sector III, Salt Lake, Kolkata 700 106, India}

\author{Saheli Mukherjee}
\email{mukherjeesaheli95@gmail.com}
\affiliation{S. N. Bose National Centre for Basic Sciences, Block JD, Sector III, Salt Lake, Kolkata 700 106, India}

\author{Bivas Mallick}
\email{bivasqic@gmail.com}
\affiliation{S. N. Bose National Centre for Basic Sciences, Block JD, Sector III, Salt Lake, Kolkata 700 106, India}
\affiliation{Statistics and Mathematics Unit, Indian Statistical Institute, 560059 Bangalore,
India}

\begin{abstract}
 Quantifying non-Markovianity is a fundamental problem in the theory of open quantum systems, as it is essential for exploiting environmental memory as a resource in quantum technologies. Existing resource-theoretic and operational quantifiers, however, generally require complete knowledge of the intermediate dynamical map and are therefore experimentally demanding. Here, we develop a witness-based framework for the quantification of non-Markovianity that bypasses complete reconstruction of the Choi state associated with the intermediate dynamical map. We first show that every non-Markovianity witness provides a certifiable lower bound on the geometric resource measure of non-Markovianity and identify the necessary and sufficient conditions under which this bound is tight. Motivated by the fact that non-positive-semidefinite witnesses are generally difficult to measure in real experimental scenario, we then focus on positive-semidefinite witnesses, which admit a natural implementation as positive operator valued measure (POVM) elements. We prove that every positive witness yields a certified lower bound on the Rivas--Huelga--Plenio (RHP) measure, thereby establishing a complete hierarchy connecting witness-based quantification, the RHP measure, and the geometric distance measure. Applying our framework to amplitude-damping and general Pauli dynamics, we find that the optimal positive witnesses reduce to simple Bell-parity projectors. Furthermore, we show that the resulting witness values can be obtained directly from a small number of experimentally accessible correlation measurements, thereby eliminating the need for complete Choi-state reconstruction and providing an efficient and experimentally feasible approach for the quantification of non-Markovianity in open quantum systems.

\end{abstract}

\maketitle

\section{Introduction}
 The theory of open quantum systems provides the fundamental framework for describing the dynamics of realistic quantum systems, where unavoidable interactions with the surrounding environment give rise to dissipation, decoherence, and irreversible evolution \cite{alicki,lindblad,gorini,breuer,RHPreview,breuerN,alonso,chruscinski2022dynamical}. Traditionally, these dynamics are described within the \emph{Markovian} approximation, where the system--environment coupling is sufficiently weak and the environment retains no memory of its interaction with the system. However, when the coupling becomes strong or the environmental correlation time is comparable to the intrinsic system timescale, this approximation breaks down, leading to the \emph{non-Markovian} regime \cite{RHPreview,breuerN,alonso,chruscinski2022dynamical,Bellomo2007,arend,Bhattacharya17,Bhattacharya20,Maity20,BBhattacharya21,wolf2008assessing,jeknic2023invertibility,PhysRevA.110.042624}. The defining feature of such dynamics is the presence of memory effects, manifested as backflow of information from the environment to the system.
 
Over the past decade, non-Markovianity has emerged as an operational resource, enabling advantages in perfect quantum teleportation with mixed states \cite{task1}, enhance the information-carrying capacity of quantum channels \cite{task2}, improving the security and
communication range of device-independent quantum secure direct
communication \cite{Roy2024QSDC}, improve entanglement distribution through noisy optical fibres \cite{task3}, increase the performance of quantum thermal machines \cite{task4}, facilitate more efficient quantum control \cite{task5}, improve metrological precision \cite{altherr2021quantum}, accelerate quantum evolution \cite{DeffnerQSL2013}, and reduce the overhead required for quantum error mitigation \cite{hakoshima2021relationship}. These advantages naturally raise a fundamental question: \emph{how can one faithfully quantify the amount of non-Markovianity generated during a quantum evolution?} Recently, considerable effort has been devoted to addressing this question, leading to several complementary approaches for characterizing memory effects in open quantum dynamics \cite{blp1,RHP,sabrina,richter2024phase,chruscinski2011measures,PhysRevA.106.042212,PhysRevA.83.062115,PhysRevA.86.044101,das2018fundamental,PhysRevA.96.062118,PhysRevA.97.032130,PhysRevA.101.020303}. In particular, recently a resource-theoretic formulation of non-Markovianity has been developed in which Markovian processes are regarded as free processes, while non-Markovian dynamics are treated as the corresponding resource. Within this framework, a geometric distance from the set of Markovian processes provides a bona fide measure of non-Markovianity \cite{Bhattacharya20}.

Alongside such distance-based quantification, two complementary paradigms have dominated the characterization of non-Markovian dynamics. The Rivas--Huelga--Plenio (RHP) approach identifies Markovianity with complete-positive (CP) divisibility of the dynamical map and quantifies memory through the failure of positivity of the intermediate Choi operator \cite{RHP,RHPreview}. In contrast, the Breuer--Laine--Piilo (BLP) approach associates non-Markovianity with revivals of the trace distance between evolving quantum states, thereby interpreting memory as information backflow \cite{BLP,blp1}. Since non-CP-divisibility is a stronger criterion than information backflow \cite{RHPreview,PhysRevA.92.042108,PhysRevA.99.042105}, we adopt the RHP framework throughout this work.

Despite the conceptual significance of these existing quantifiers, evaluating these measures is computationally and experimentally demanding. The RHP measure requires evaluating the trace norm of the intermediate Choi operator, which in turn demands complete knowledge of its spectrum \cite{RHPreview}. Likewise, the distance measure of non-Markovainity involves an optimization over the entire convex set of Markovian Choi states \cite{Bhattacharya20}. Therefore, in practice, both quantities are typically obtained through complete Choi state reconstruction, which can be experimentally and computationally demanding, particularly as the system size increases \cite{PhysRevLett.109.120403}. In this work, we overcome this limitation by developing a witness-based framework that quantifies non-Markovianity directly from witness expectation values, without reconstructing the complete process. We first establish that every non-Markovianity witness, detecting a given evolution provides a lower bound on the geometric distance measure of non-Markovianity. We then characterize precisely when this bound becomes tight, showing that saturation occurs if and only if the witness possesses an eigenspace structure matching the positive and negative spectral subspaces of the difference between the non-Markovian Choi operator and its closest Markovian approximation. Although this characterization identifies the optimal witnesses, such witnesses may contain both positive and negative eigenvalues, which makes their direct experimental implementation challenging.

This motivates our discussion on the class of positive semidefinite non-Markovianity witnesses. After suitable normalization, they constitute elements of a positive operator-valued measure (POVM), so that their expectation values correspond directly to experimentally observable probabilities. We prove that every positive non-Markovianity witness yields a certified lower bound on the RHP measure. Consequently, our results establish a complete quantitative hierarchy connecting experimentally measurable witness expectations, the RHP measure, and the geometric distance measure of non-Markovianity. To demonstrate the practical utility of our framework, we then apply our approach to two different paradigmatic qubit dynamical models: amplitude damping and general Pauli dynamics. For all these models, the optimal positive witnesses reduce to simple Bell-parity projectors. We then show that for amplitude damping, a single parity witness exactly reproduces the RHP measure, while for general Pauli dynamics the complete measure is recovered from only three parity witnesses. Moreover, we propose an explicit experimental scheme for directly
reconstructing the witness values from a small set of two-qubit
Pauli-parity correlations measured along the physical evolution,
thereby avoiding any implementation or reconstruction of the
intermediate Choi operator. The required parity measurements are
compatible with established experimental platforms, including
superconducting circuits~\cite{Riste2013,Steffen2012},
polarization-encoded photonic systems~\cite{James2001,Barbieri2003},
and trapped-ion architectures~\cite{Leibfried2003,Bruzewicz2019}. Consequently, our approach reduces the experimental overhead from complete process characterization to only one or three correlation measurements, together with their time derivatives for these canonical models, thus providing an efficient, scalable, and operationally accessible framework for quantifying non-Markovianity.

The paper is organized as follows. In Sec. \ref{s2}, we introduce the necessary preliminaries, including the dynamics of open quantum systems, witness-based detection of non-Markovianity, and the RHP and distance-based measures. In Sec. \ref{s3}, we develop our witness-based framework for the quantitative characterization of non-Markovian dynamics and establish the hierarchy connecting witness-based quantification with the RHP and geometric measures. We then apply our framework to amplitude-damping, and general Pauli dynamics, showing that optimal positive witnesses can be realized as simple Bell-parity projectors in Sec. \ref{Sec:Parity_Section}. We further demonstrate how the corresponding witness values can be reconstructed from experimentally accessible two-qubit correlation measurements. Finally, Sec. \ref{Conclusion} summarizes our main results and discusses their operational significance.

\section{Preliminaries}\label{s2}
We begin by providing a concise background of the basic concepts and mathematical tools used throughout this work. We adopt the standard formalism and terminology of quantum information theory unless stated otherwise.
\begin{table*}[t]
	\centering
	\caption{Witness-related notation used throughout the manuscript.}
	\label{tab:witness_notation}
	
	\renewcommand{\arraystretch}{1.3}
	\begin{tabular}{ll}
		\toprule
		\textbf{Notation} & \textbf{Meaning} \\
		\midrule
		
		$\mathcal W$
		&
		Generic Hermitian non-Markovianity witness.
		\\
		
		$\mathfrak W_{\mathrm H}$
		&
		Set of Hermitian non-Markovianity witnesses.
		\\
		
		$\mathcal W^{+}$
		&
		Positive-semidefinite non-Markovianity witness.
		\\
		
		$\mathfrak W^{+}$
		&
		Set of positive-semidefinite non-Markovianity witnesses.
		\\
		
		$\mathcal W_c$
		&
		Rescaled witness,
		$\displaystyle
		\mathcal W_c=
		\frac{2\mathcal W}{\lambda_{+}-\lambda_{-}}$.
		\\
		
		$\mathcal W'$
		&
		Shifted normalized operator,
		$\displaystyle
		\mathcal W'=\mathcal W_c-\alpha\mathbb I$,
		$\|\mathcal W'\|_{\infty}=1$,
		$\displaystyle
		\alpha=
		\frac{\lambda_{+}+\lambda_{-}}
		{\lambda_{+}-\lambda_{-}}$.
		\\
		
		$\widetilde w_c$
		&
		Witness rate associated with the rescaled witness $\mathcal W_c$.
		\\
		
		$\mathcal W_c^{+}$
		&
		Rescaled positive witness,
		$\displaystyle
		\mathcal W_c^{+}
		=
		\frac{2\mathcal W^{+}}
		{\lambda_{+}-\lambda_{-}}$.
		\\
		
		$\widetilde w_c^{+}$
		&
		Witness rate associated with the positive witness $\mathcal W_c^{+}$.
		\\
		
		$\mathcal W^{+}(\epsilon)$
		&
		$\epsilon$-adapted positive witness used in the RHP-rate optimization.
		\\
		
		$W_{c,i}$
		&
		Pauli-parity witness,
		$\displaystyle
		W_{c,i}=\mathbb I-\sigma_i\otimes\sigma_i$,
		$i\in\{x,y,z\}$.
		\\
		
		$\widetilde w_i$
		&
		Scalar witness rate associated with $W_{c,i}$,
		$i\in\{x,y,z\}$.
		\\
		
		$W_{\mathrm{Amp}}^{\mathrm{opt}}$
		&
		Optimal positive witness for the amplitude-damping model.
		\\
		
		$W_{\mathrm{Th}}^{\mathrm{opt}}$
		&
		Optimal positive witness for the thermal model.
		\\
		
		$W_{\mathrm{Pauli}}^{\mathrm{opt}}$
		&
		Optimal positive witness for the general Pauli model.
		\\
		
		\bottomrule
	\end{tabular}
\end{table*}
\subsection{Dynamics of Open Quantum System}
An isolated quantum system evolves unitarily. In contrast, the evolution of an open quantum system is described by a completely positive trace-preserving (CPTP) dynamical map $\Lambda (t,t_0)$,  which transforms an initial density operator $\rho (t_0) \in \mathcal{D}(\mathcal{H})$ onto the evolved state $\rho (t) \in \mathcal{D}(\mathcal{H})$ according to
\begin{equation*}
    \Lambda (t,t_0) : \rho (t_0) \mapsto \rho (t).
\end{equation*}
$\mathcal{D}(\mathcal{H})$ represent the set of density operators acting on the Hilbert space $\mathcal{H}$. Throughout this work, we assume that the inverse map $\Lambda^{-1} (t,t_0)$ exists for all times between $t_0$ and $t (t > t_0)$. Consequently, for any intermediate time $t \geq \tau \geq t_0 $, the evolution admits the decomposition
\begin{equation}
\Lambda (t,t_0)= \Lambda (t,\tau) \circ \Lambda (\tau,t_0).  \label{divisibility}
\end{equation}
where $\circ$ denotes the composition of the two dynamical maps.

Although both $\Lambda (t,t_0)$ and $\Lambda (\tau,t_0)$ are CPTP, the intermediate map $\Lambda (t,\tau)$ is not necessarily completely positive. This distinction gives rise to the notion of 
divisibility. A dynamical map is said to be \textit{divisible} if the above composition law given by Eq.~\eqref{divisibility} holds for every $t \geq \tau \geq t_0 $. It is called \textit{positive divisible} (P-divisible) when the intermediate map $\Lambda (t,\tau)$ is positive for all such times, and satisfies the divisible property. The dynamics is \textit{completely positive divisible} (CP-divisible) when the divisible property is satisfied, and $\Lambda (t,\tau)$ is itself CPTP for every $t \geq \tau \geq t_0$. 

The mathematical characterization of a dynamical map in terms of divisibility provides a natural framework for defining quantum Markovianity. Intuitively, a memoryless evolution can be decomposed into a sequence of physically valid intermediate maps satisfying the composition law. Within this framework, Rivas, Huelga, and Plenio (RHP) defined a quantum dynamics to be Markovian if it is completely positive divisible (CP-divisible), and non-Markovian otherwise \cite{RHP}. Accordingly, the breakdown of CP divisibility serves as the hallmark of non-Markovian memory effects. An alternative characterization was proposed by Breuer, Laine, and Piilo (BLP), based on the distinguishability of quantum states under dynamical evolution \cite{BLP}. Owing to the interaction with the environment, the distinguishability between two quantum states generally decreases monotonically with time. However, any temporary increase in distinguishability signifies a backflow of information from the environment to the system, thereby indicating non-Markovian behavior. The former notion is commonly referred to as RHP-type non-Markovianity, whereas the latter is known as BLP-type non-Markovianity. It is well known that every CP-divisible dynamics is also Markovian in the BLP sense, whereas the converse does not hold in general \cite{Hall2014}. Consequently, information backflow constitutes a weaker criterion than CP divisibility. Throughout this work, we adopt CP divisibility as the defining criterion of quantum Markovianity, and any deviation from CP divisibility is regarded as non-Markovian. Also, we restrict our attention to quantum dynamical maps generated by Lindblad-type master equations. The evolution of the system density operator $\rho$ is governed by
\begin{equation}
    \frac{d\rho}{dt} = -\frac{i}{\hbar} [H, \rho] + \sum_{k} \gamma_k (L_k \rho {L_k}^{\dagger} - \frac{1}{2}({L_k}^{\dagger}  L_k \rho + \rho {L_k}^{\dagger}  L_k ))
\end{equation}
where $H$ is the system Hamiltonian describing the unitary part of the evolution, $\gamma_k$ are the Lindblad decay rates, and $L_k$ denote the Lindblad operators characterizing the dissipative interaction with the environment \cite{lindblad}.

Every quantum dynamical map $\Lambda(t,\tau)$ admits a one-to-one correspondence with a quantum state through the Choi--Jamiołkowski isomorphism \cite{choi,jamil}. The corresponding Choi state, denoted by $\mathcal{C}(t,\tau)$ is defined as \cite{choi}
\begin{equation}
\mathcal{C}(t,\tau)
=
(\mathbb{I}\otimes\Lambda(t,\tau))
\ket{\phi}\bra{\phi},
\label{choi}
\end{equation}
where $\ket{\phi}$ is a maximally entangled bipartite state in a $d\times d$-dimensional Hilbert space. The Choi--Jamiołkowski isomorphism establishes that a dynamical map $\Lambda(t,\tau)$ is completely positive if and only if its corresponding Choi state $\mathcal{C}(t,\tau)$ is positive semidefinite. Consequently, the complete positivity of a quantum channel can be verified entirely through the positivity of its associated Choi state.
\subsection{Witness-based detection of non-Markovianity}
In this section, we review the detection of non-Markovian dynamics using witness operators. 

Let $\mathbb{F}_{\mathrm{M}}^{\epsilon}$ denote the set of Choi operators corresponding to all Markovian dynamics. It has been shown in \cite{Bhattacharya20} that $\mathbb{F}_{\mathrm{M}}^{\epsilon}$ is convex and compact. As a consequence, the Hahn-Banach hyperplane separation theorem \cite{holmes2012geometric} proposes the existence of a hermitian operator (popularly known as the \textit{witness} operator) that separates all Markovian dynamics $\in \mathbb{F}_{\mathrm{M}}^{\epsilon}$ from at least one non-Markovian dynamics through a hyperplane. Mathematically, if $\mathcal{W}$ represents a non-Markovianity witness, then 
\begin{equation}
\begin{split}
    & \mathrm{Tr}
\big(
\mathcal{W}\mathcal{C}^{\mathrm{M}}
\big)
\geq 0,
\qquad
\forall~
\mathcal{C}^{\mathrm{M}}
\in
\mathbb{F}_{\mathrm{M}}^{\epsilon} \hspace{0.2cm} \text{and} \\ &
\mathrm{Tr}
\big(
\mathcal{W}\mathcal{C}^{\mathrm{N}}
\big)
< 0 \hspace{0.2cm}\text{for at least one} \hspace{0.2cm}\mathcal{C}^{\mathrm{N}} \notin \mathbb{F}_{\mathrm{M}}^{\epsilon}.
\end{split} \label{witness}
\end{equation}
Henceforth, the set of all Hermitian witnesses satisfying Eq.~\eqref{witness} is represented by $\{\mathfrak{W}_{\mathrm H}\}$, i.e., 
\begin{equation}
    \mathfrak{W}_{\mathrm H}=
    \left\{
    \mathcal{W}: \mathcal{W}=\mathcal{W}^{\dagger} \text{and} 
    \operatorname{Tr}(\mathcal{C}^\mathrm{M}\mathcal{W})\geq0,\ 
    \forall\,\mathcal{C}^\mathrm{M}\in \mathbb{F}_{\mathrm{M}}^{\epsilon}
    \right\}.
\end{equation}
The Choi operators corresponding to a Markovian and non-Markovian dynamics are represented by $\mathcal{C}^{\mathrm{M}}$ and $\mathcal{C}^{\mathrm{N}}$, respectively. Thus, a negative expectation value of the witness operator is sufficient to conclude that the underlying dynamics is non-Markovian. It is worth noting that witness-based detection of memory effects has been demonstrated in several realistic physical models \cite{giarmatzi2021witnessing,PhysRevA.102.022221}, highlighting the practical applicability and operational significance of this approach.

Beyond mere detection, a question of fundamental interest is to quantify the amount of non-Markovianity present in a given dynamics. In this context, we discuss the Rivas--Huelga--Plenio (RHP) and distance-based measures of non-Markovian dynamics in the following section.  
\subsection{RHP and Distance-based measures of non-Markovianity}
The RHP measure quantifies the extent to which the Choi operator associated with a non-Markovian dynamics deviates from positivity. It is formally defined as \cite{RHP}
\begin{equation}
    \mathcal{I}_{\text{RHP}}= \int_{0}^{\infty} g(t) \,dt, 
\end{equation}
where 
\begin{equation}
    g(t)
    =
    \lim_{\epsilon\to0^+}
    \frac{
        \|\mathcal{C}(t+\epsilon,t)\|_1-1
    }{\epsilon}.   \label{RHP}
\end{equation}
For a Markovian dynamics, $\mathcal{C}(t+\epsilon,t) \in \mathbb{F}_{\mathrm{M}}^{\epsilon}$ and hence, $g(t)=\mathcal{I}_{\text{RHP}}=0$ since $\|\mathcal{C}(t+\epsilon,t)\|_1=1$. On the other hand, for $\mathcal{C}(t+\epsilon,t) \notin \mathbb{F}_{\mathrm{M}}^{\epsilon}$, $\mathcal{I}_{\text{RHP}}>0$ and this is a signature of non-Markovianity. While the RHP measure provides a mathematically rigorous characterization, a more intuitive interpretation can be obtained through a distance-based measure.

Distance-based measures in resource theories quantify the minimum geometric distance of the resourceful state from the set of free states. A resource theory of non-Markovianity has been formulated in \cite{Bhattacharya20} where the Choi state corresponding to Markovian dynamics are free states, and the divisible operations are free. Non-Markovian dynamics are the resources in this theory. The set of divisible operations does not form a convex set, in general. To overcome this, the small-time approximation is taken into account to ensure the convexity of free operations. For the resource theory of non-Markovianity, the distance-based measure is defined as
\begin{equation}
    D_T(t)
=
\lim_{\epsilon\to0^+}
\inf_{\mathcal{C}^{\mathrm{M}} \in \mathbb{F}_{\mathrm{M}}^{\epsilon}}
\frac{
\left\|
\mathcal{C}^{\mathrm{N}}(t+\epsilon,t)
-
\mathcal{C}^{\mathrm{M}}(t+\epsilon,t)
\right\|_1
}{\epsilon}, \label{distance}
\end{equation}
where the trace norm of an operator $A$ is given by $\|A\|_1=\operatorname{Tr}\sqrt{A^{\dagger}A}$.  $\mathcal{C}^{\mathrm{M}} \in \mathbb{F}_{\mathrm{M}}^{\epsilon}$ and the small time approximation is ensured by the limiting condition $\epsilon\to0^+$. This measure is faithful in the sense that $D_T(t) \ge 0$ with the equality satisfied if and only if the dynamics is Markovian. In addition, the measure also obeys convexity and monotonicity under divisible free operations. So, $D_T(t)$ represents a \textit{bona fide} measure of non-Markovianity \cite{Bhattacharya20}.
\begin{observation} \cite{Bhattacharya20}
    The distance-based measure is related to the RHP measure as follows:
    \begin{equation*}
        D_{T}(t) \ge g(t), 
    \end{equation*}
    i.e., the RHP measure is a lower bound to the distance-based measure. \label{observation1}
\end{observation} 

While distance-based measures are particularly appealing owing to their clear geometric interpretation, their evaluation requires a minimization over the set of free states, which is difficult to compute exactly. To overcome this difficulty, we derive a lower bound on the distance measure in the following section.

\section{Quantification of non-Markovian dynamics} \label{s3}
Having introduced witness-based detection and distance-based quantification of non-Markovianity, we now unify these two frameworks by deriving a witness-based lower bound on the distance measure.
\subsection{Witness-based lower bound of non-Markovianity}
In this section, we establish the central result of our work through the following theorem, which relates the expectation value of a suitably chosen normalized non-Markovianity witness to the geometric measure of non-Markovianity. However, before proceeding to the main result, we formally define the normalized witness below.
\begin{definition}\label{def1}
Let $\mathcal W$ be a non-Markovianity witness satisfying
\begin{align*}
\operatorname{Tr}(\mathcal{C}^{\mathrm{M}}\mathcal W)\ge0,
\qquad
\forall\, \mathcal{C}^{\mathrm{M}}\in \mathbb{F}_{\mathrm{M}}^{\epsilon},  
\end{align*}

where $\mathbb{F}_{\mathrm{M}}^{\epsilon}$ denotes the set of Markovian Choi operators.
Let $\lambda_+$ and $\lambda_-$ denote respectively the largest and
smallest eigenvalues of $\mathcal W$, with $\lambda_+>\lambda_-$.
The associated normalized witness operator is defined as
\begin{equation}
\mathcal W'
=
\frac{
2\mathcal W-(\lambda_++\lambda_-)\mathbb I
}{
\lambda_+-\lambda_-
}
=
\mathcal W_c-\alpha\mathbb I,   \label{wprime}
\end{equation}  
where
\begin{align*}
  \mathcal W_c
:=
\frac{2\mathcal W}{\lambda_+-\lambda_-},
\qquad
\alpha
:=
\frac{\lambda_++\lambda_-}{\lambda_+-\lambda_-}.  
\end{align*}

By construction,
\begin{align*}
  -\mathbb I
\le
\mathcal W'
\le
\mathbb I,  
\end{align*}

or equivalently,
\begin{align*}
   \|\mathcal W'\|_\infty=1. 
\end{align*}

\end{definition}
We now present our main theorem, which establishes a witness-based lower bound on the geometric distance measure of non-Markovianity.
\begin{theorem}\label{theorem1}
For every non-Markovianity witness ($\mathcal{W} \in {\mathfrak{W}_{\mathrm H}}$) detecting a non-Markovian dynamics,
\begin{equation}
D_T(t)\ge \widetilde w_c,
\label{eq:mainbound}
\end{equation}
where
\begin{align} \label{widetildewc1}
  \widetilde w_c
=
-\lim_{\epsilon\to0^+}
\frac{
\operatorname{Tr}
\left[
\mathcal{C}^\mathrm{N}(t+\epsilon,t)\mathcal W_c
\right]
}{\epsilon},  
\end{align}

and $D_T(t)$ is the trace-distance measure defined in
Eq.~\eqref{distance}.
\end{theorem}

\begin{proof}

Consider the trace-norm identity for operators $A$ and $X$ such that $-\mathbb{I}\le X\le \mathbb{I}$.
\begin{align}
    \|A\|_1
=
\max_{-\mathbb{I}\le X\le \mathbb{I}}
\left|
\mathrm{Tr}(AX)
\right|.  \label{tracenormidentity}
\end{align}

By construction, the spectrum of $\mathcal{W}'$ lies in the interval $[-1,1]$, and therefore
\begin{align*}
  -\mathbb{I}
\leq
\mathcal{W}'
\leq
\mathbb{I}.  
\end{align*}

Using the above identity with $X=\mathcal{W}'$ and $A= \mathcal{C}^{\mathrm{N}}(t+\epsilon,t)-\mathcal{C}^{\mathrm{M}}(t+\epsilon,t)$, we obtain,
\begin{align}
\left\|
\mathcal{C}^{\mathrm{N}}(t+\epsilon,t)
-
\mathcal{C}^{\mathrm{M}}(t+\epsilon,t)
\right\|_1
\ge
\Big|
\mathrm{Tr}
\Big[\big(
\mathcal{C}^{\mathrm{N}}(t+\epsilon,t)
-\nonumber
\\
\mathcal{C}^{\mathrm{M}}(t+\epsilon,t)\big)
\mathcal{W}'
\Big]
\Big|.
\end{align}

Substituting $D_T(t)$ from Eq.~\eqref{distance} gives 
\begin{align*}
   D_T(t)
\geq
\lim_{\epsilon\to0^+}
\inf_{\mathcal{C}^{\mathrm{M}}}
\frac{
\left|
\mathrm{Tr}
\left[
\left(
\mathcal{C}^{\mathrm{N}}(t+\epsilon,t)
-
\mathcal{C}^{\mathrm{M}}(t+\epsilon,t)
\right)
\mathcal{W}'
\right]
\right|
}{\epsilon}. 
\end{align*}

Substituting the definition of $\mathcal{W}'$ from Definition \ref{def1} and using the normalization condition $\mathrm{Tr}(\mathcal{C}^\mathrm{M})=\mathrm{Tr}(\mathcal{C}^\mathrm{N})=1$ for a trace-preserving dynamics gives
\begin{align*}
    D_T(t)
\geq
\lim_{\epsilon\to0^+}
\inf_{\mathcal{C}^{\mathrm{M}}}
\frac{
\left|
\mathrm{Tr}
\left[
\left(
\mathcal{C}^{\mathrm{N}}(t+\epsilon,t)
-
\mathcal{C}^{\mathrm{M}}(t+\epsilon,t)
\right)
\mathcal{W}_{c}
\right]
\right|
}{\epsilon}.
\end{align*}

Let $\mathcal{C}^{\mathrm{M}}_{\mathrm{opt}}$ denote the Markovian Choi operator attaining the infimum. Since $\mathcal{W}$ is a non-Markovianity witness,
$$
\mathrm{Tr}
\left(
\mathcal{C}^{\mathrm{M}}_{\mathrm{opt}}(t+\epsilon,t)
\mathcal{W}
\right)
\geq 0
\implies
\mathrm{Tr}
\left(
\mathcal{C}^{\mathrm{M}}_{\mathrm{opt}}(t+\epsilon,t)
\mathcal{W}_{c}
\right)
\geq 0,
$$whereas for the detected non-Markovian Choi operator
$\mathcal{C}^{\mathrm N}(t+\epsilon,t)$,
the witness satisfies
$$
\mathrm{Tr}
\left(
\mathcal{C}^{\mathrm N}(t+\epsilon,t)\mathcal W
\right)
<0
\implies
\mathrm{Tr}
\left(
\mathcal{C}^{\mathrm N}(t+\epsilon,t)\mathcal W_c
\right)
<0.
$$

Consequently,
\begin{equation*}
\begin{split}
&\mathrm{Tr}
\left[
\left(
\mathcal{C}^{\mathrm{N}}(t+\epsilon,t)
-
\mathcal{C}^{\mathrm{M}}_{\mathrm{opt}}(t+\epsilon,t)
\right)
\mathcal{W}_{c}
\right]
\\ &=
\mathrm{Tr}
\left(
\mathcal{C}^{\mathrm{N}}(t+\epsilon,t)
\mathcal{W}_{c}
\right)
-
\mathrm{Tr}
\left(
\mathcal{C}^{\mathrm{M}}_{\mathrm{opt}}(t+\epsilon,t)
\mathcal{W}_{c}
\right)
\\ & \le
\mathrm{Tr}
\left(
\mathcal{C}^{\mathrm{N}}(t+\epsilon,t)
\mathcal{W}_{c}
\right)
< 0.
\end{split}
\end{equation*}
Using the properties of a non-Markovian witness,
\begin{equation*}
\begin{split}
    & \left|
\mathrm{Tr}
\left[
\left(
\mathcal{C}^{\mathrm{N}}(t+\epsilon,t)
-
\mathcal{C}^{\mathrm{M}}_{\mathrm{opt}}(t+\epsilon,t)
\right)
\mathcal{W}_{c}
\right]
\right|
\\ & \ge
\left|
\mathrm{Tr}
\left(
\mathcal{C}^{\mathrm{N}}(t+\epsilon,t)
\mathcal{W}_{c}
\right)
\right|.
\end{split}
\end{equation*}

Hence,
\begin{equation*}
    D_T(t)
\geq
\lim_{\epsilon\to0^+}
\frac{
\left|
\mathrm{Tr}
\left[
\mathcal{C}^{\mathrm{N}}(t+\epsilon,t)
\mathcal{W}_{c}
\right]
\right|
}{\epsilon}.
\end{equation*}

Now, since
\begin{equation*}
\begin{split}
& \mathrm{Tr}
\left(
\mathcal{C}^{\mathrm{N}}(t+\epsilon,t)
\mathcal{W}_{c}
\right)
<0, \\ &
D_T(t)
\geq
-
\lim_{\epsilon\to0^+}
\frac{
\mathrm{Tr}
\left[
\mathcal{C}^{\mathrm{N}}(t+\epsilon,t)
\mathcal{W}_{c}
\right]
}{\epsilon}
=
\widetilde{w}_{c}.
\end{split}
\end{equation*}

This completes the proof.
\end{proof}

Theorem~\ref {theorem1} establishes that whenever a non-Markovian dynamics is detected by a witness, the corresponding geometric distance measure is necessarily non-zero, thereby certifying the presence of non-Markovianity in the underlying dynamics. Moreover, unlike distance-based measures, whose exact evaluation generally requires a minimization over the set of free states, the proposed witness-based framework provides a computable lower bound that circumvents this optimization. We next characterize the conditions under which this lower bound is saturated, thereby identifying when the witness is optimal.


\begin{proposition}[Tightness criterion]
\label{prop:tight}
Let
\[
\Delta_\epsilon
=
\mathcal{C}^N(t+\epsilon,t)
-
\mathcal{C}_{\rm opt}^{\rm M}(t+\epsilon,t),
\]
where $\mathcal{C}_{\rm opt}^{\rm M}$ denotes the Markovian Choi operator attaining
the infimum in Eq.~\eqref{distance}. Let
\begin{align*}
    \Delta_\epsilon
=
\Delta_\epsilon^+
-
\Delta_\epsilon^-
\end{align*}

be the Jordan decomposition of $\Delta_\epsilon$ into orthogonal
positive operators.
Then
\begin{align*}
    \widetilde w_c=D_T(t)
\end{align*}

if and only if, in the limit $\epsilon\to0^+$, the following
conditions hold:

\begin{enumerate}
\item[\rm (i)]
\emph{(Vanishing Markovian contribution)}
\begin{align*}
\operatorname{Tr}
\left[
\mathcal{C}_{\rm opt}^{\rm M}(t+\epsilon,t)\mathcal W_c
\right]
=
0.  
\end{align*}

\item[\rm (ii)]
\emph{(Sign matching)}
$$\mathcal W'
=
-\operatorname{sign}(\Delta_\epsilon)
\qquad
\text{on }
\operatorname{supp}(\Delta_\epsilon).$$

\end{enumerate}
\end{proposition}

\begin{proof}
Let
$$\Delta_\epsilon
=
\mathcal{C}^{\rm N}(t+\epsilon,t)-\mathcal{C}_{\rm opt}^{\rm M}(t+\epsilon,t).$$

$$ \implies \mathcal{C}^{\rm N}(t+\epsilon,t)
=
\mathcal{C}_{\rm opt}^{\rm M}(t+\epsilon,t)+\Delta_\epsilon.$$

Therefore, 
\begin{align*}
   \widetilde w_c
=
-
\lim_{\epsilon\to0^+}
\frac{
\operatorname{Tr}
\left[
\mathcal{C}_{\rm opt}^{\rm M}(t+\epsilon,t)\mathcal W_c
\right]
+
\operatorname{Tr}
\left[
\Delta_\epsilon\mathcal W_c
\right]
}{\epsilon}. 
\end{align*}

\vspace{0.5cm}

Let
\[
\widetilde w_c
=
-A-B,
\]
where
\[
A
:=
\lim_{\epsilon\to0^+}
\frac{
\operatorname{Tr}
\left[
\mathcal{C}_{\rm opt}^{\rm M}(t+\epsilon,t)\mathcal W_c
\right]
}{\epsilon},
\]
and
\begin{equation*}
B
:=
\lim_{\epsilon\to0^+}
\frac{
\operatorname{Tr}
\left[
\Delta_\epsilon\mathcal W_c
\right]
}{\epsilon}
\end{equation*}
Using Eq.~\eqref{wprime} and the traceless property of $\Delta_\epsilon$ for a trace-preserving dynamics, $B$ can be equivalently written as 
\begin{equation}
    B:=
\lim_{\epsilon\to0^+}
\frac{
\operatorname{Tr}
\left[
\Delta_\epsilon\mathcal W'
\right]
}{\epsilon}.  \label{Bvalue}
\end{equation}

Since \(\mathcal W\) is a non-Markovianity witness, $\operatorname{Tr}
(\mathcal{C}_{\rm opt}^{\rm M}\mathcal W)\ge0$. For the sake of simplicity, we use the notation $\mathcal{C}_{\rm opt}^{\rm M}$ instead of $\mathcal{C}_{\rm opt}^M(t+\epsilon,t)$ henceforth.

Since
\begin{align*}
 \mathcal W_c
=
\frac{2}{\lambda_+-\lambda_-}\mathcal W,    
\end{align*}

it follows that
\begin{align*}
\operatorname{Tr}
(\mathcal{C}_{\rm opt}^{\rm M}\mathcal W_c)\ge0,
\end{align*}
and therefore $A\ge0$. Furthermore, Definition~\ref{def1} guarantees
$\|\mathcal W'\|_\infty=1.$

From Eq.~\eqref{tracenormidentity}, we have
\begin{align*}
\operatorname{Tr}(\Delta_\epsilon \mathcal{W}')
&\ge
-\|\Delta_\epsilon\|_1.
\end{align*}
Therefore,
\begin{align*}
B
&\ge
-D_T(t).
\end{align*}

Assume that conditions (i) and (ii) hold.
Condition (i) gives $A=0$. Using condition (ii), 
\begin{align*}
    \mathcal W'
=
-\operatorname{sign}(\Delta_\epsilon)
\qquad
\text{on }
\operatorname{supp}(\Delta_\epsilon),
\end{align*}
or equivalently,
\begin{align*}
   \mathcal W'
=
-\mathbb I
\quad
\text{on }
\operatorname{supp}(\Delta_\epsilon^+),
\qquad
\mathcal W'
=
+\mathbb I
\quad
\text{on }
\operatorname{supp}(\Delta_\epsilon^-). 
\end{align*}
The maximum of Eq.~\eqref{tracenormidentity} is attained for these extreme values of $\mathcal W'$. So for $X=\mathcal W'$,
$\|\Delta_\epsilon\|_1=|\operatorname{Tr}
(\Delta_\epsilon\mathcal W')|$. Now,
$$\operatorname{Tr}(\Delta_\epsilon\mathcal W') \le  \operatorname{Tr}[\mathcal{C}^{\mathrm{N}}(t+\epsilon,t)\mathcal W'] < 0.$$
$$ \implies |\operatorname{Tr}(\Delta_\epsilon\mathcal W')| = - \operatorname{Tr}(\Delta_\epsilon\mathcal W').$$ Hence, 
\begin{align*}
    -\|\Delta_\epsilon\|_1=-|\operatorname{Tr}
(\Delta_\epsilon\mathcal W')|=\operatorname{Tr}
(\Delta_\epsilon\mathcal W').
\end{align*}

$\implies -D_T(t)=B$, and
\begin{align*}
 \widetilde w_c
=
-A-B
=-B=
D_T(t).   
\end{align*}

Now, we prove the converse statement. Precisely, if $\widetilde w_c=D_T(t)$, then conditions (i) and (ii) hold. Let
$$\widetilde w_c = -A-B.$$ 
Since
\begin{align*}
    A\ge0 \qquad  \text{and}
\qquad 
B\ge-D_T(t),
\end{align*}

\begin{align*}
    \widetilde w_c=-A-B=D_T(t)
\end{align*}

implies that
\begin{align*}
    A=0, \text{which is condition (i).}
\end{align*}
 Thus, $\widetilde w_c=-B=D_T(t).$
From Eqs.~\eqref{distance} and \eqref{Bvalue}, this implies that 
\begin{align*}
  \operatorname{Tr}
(\Delta_\epsilon\mathcal W')
=
-\|\Delta_\epsilon\|_1, 
\end{align*}

which holds if and only if
\begin{align*}
    \mathcal W'
=
-\operatorname{sign}(\Delta_\epsilon)
\end{align*}

on
$\operatorname{supp}(\Delta_\epsilon)$,
which establishes condition (ii). Hence, the converse holds. In other words, $\widetilde w_c=D_T(t)$ is obtained iff (i) and (ii) hold.
\end{proof}

The tightness criterion of Proposition \ref{prop:tight} generally requires witness operators with both positive and
negative eigenvalues. In particular, an optimal witness must mirror the
sign structure of $\Delta_\epsilon$, i.e., assign  
an eigenvalue $-1 (+1)$ on the positive (negative) eigenspace of
$\Delta_\epsilon$.

Observation \ref{observation1} and Theorem \ref{theorem1} give a lower bound to $D_T(t)$ in terms of the RHP measure $g(t)$ and $\widetilde w_c$, respectively. Therefore, a question of fundamental interest is a relation between $g(t)$ and $\widetilde w_c$. Below (in Theorem \ref{thm:optimal}), we provide a relation between them for positive witnesses. In contrast to entanglement theory, where witnesses are necessarily Hermitian and non-positive in order to satisfy the condition $\operatorname{Tr}(\rho W)<0$ for an entangled state $\rho$, the resource theory of non-Markovianity permits positive witnesses. This is because the Choi state associated with a non-Markovian dynamics is non-positive, allowing the condition in Eq.~\eqref{witness} to be satisfied even when the witness itself is positive. Moreover, positive witnesses are of special interest from an experimental perspective, since they correspond to measurement effects and can therefore be implemented as expectation values of positive-operator-valued-measures (POVMs) acting on the intermediate Choi operator~\cite{Bhattacharya20}. The set of all such positive witnesses is represented by $\{\mathfrak{W}^{+}\}$, i.e., 
\begin{equation}
    \mathfrak{W}^{+}=
    \left\{
    \mathcal{W}^{+}\geq0:
    \operatorname{Tr}(\mathcal{C}^{\mathrm{M}}\mathcal{W}^{+})\geq0,\ 
    \forall\,\mathcal{C}^{\mathrm{M}}\in \mathbb{F}_{\mathrm{M}}^{\epsilon}
    \right\}.\\
    \end{equation}
    The significance of using positive witnesses for deriving a relation between $g(t)$ and $\widetilde w_c$ is further illustrated through the following theorem.


\begin{theorem}
\label{thm:optimal}
If
\begin{equation}  \label{widetildewc2}
\widetilde w_c^+
=
-\lim_{\epsilon\to0^+}
\frac{
\operatorname{Tr}
\left[
\mathcal{C}(t+\epsilon,t)\mathcal W_c^+
\right]
}{\epsilon},
\end{equation}
where $\mathcal W_c^+ = \frac{2\mathcal W^+}{\lambda_+-\lambda_-}$ and $\mathcal W^{+} \in \mathfrak{W}^{+}$ represents a positive witness, then for
every such positive witness,
\[
\widetilde w_c^+
\le
g(t),
\]
where $g(t)$ is the RHP non-Markovianity rate given by
Eq.~\eqref{RHP}. The equality is attained by the $\epsilon$-adapted
family of positive witnesses
\[
\mathcal W^+(\epsilon)=P_-(\epsilon),
\]
where $P_-(\epsilon)$ denotes the projector onto the support of
$\mathcal C_\epsilon^-$. Consequently,
\begin{align*}
    g(t)
    &=
    \sup_{\{\mathcal W^+(\epsilon) \in  \mathfrak{W}^{+}\}}
    \left[
        \lim_{\epsilon\to0^+}
        \frac{
            -\operatorname{Tr}
            \Big(
                \mathcal{C}(t+\epsilon,t)\mathcal W_c^+(\epsilon)
            \Big)
        }{\epsilon}
    \right] \\
    &\le
    \sup_{\mathcal W \in  \mathfrak{W}_{\mathrm H}}
        \widetilde w_c
    \le
    D_T(t).
\end{align*}
Here the first supremum ranges over all $\epsilon$-adapted families of
positive witnesses, while the second supremum is taken over all
Hermitian witnesses satisfying
$\operatorname{Tr}(\mathcal{C}^{\mathrm{M}}\mathcal W)\ge0,
\forall\, \mathcal{C}^M\in\mathbb{F}_{\rm M}^{\epsilon}$.

If, in addition, $\operatorname{supp}(\mathcal{C}_\epsilon^-)
=
\operatorname{supp}(\mathcal{C}_{\epsilon'}^-)$ for all sufficiently small $\epsilon,\epsilon'>0$, then there exists a
fixed projector $P_-$, such that
$P_-(\epsilon)=P_-$
for all sufficiently small $\epsilon$, and therefore
\begin{align*}
    g(t)
    &=
    \max_{\mathcal W^+ \in {\mathfrak{W}^+}}
    \left[
        \widetilde w_c^+
    \right] \\
    &\le
    \sup_{\mathcal W \in {\mathfrak{W}_{\mathrm H}}}
    \left[
        \widetilde w_c
    \right]
    \le
    D_T(t).
\end{align*}

\end{theorem}

\begin{proof}
Let
\begin{align*}
    \mathcal{C}(t+\epsilon,t)
=
\mathcal{C}_{\epsilon}^+-\mathcal{C}_{\epsilon}^-
\end{align*}

be the Jordan decomposition of the intermediate Choi operator.

Since $\mathcal{C}_{\epsilon}^+ \ge 0$ and $\mathcal{C}_{\epsilon}^+ \mathcal{C}_{\epsilon}^-= 0$, 
\begin{align*}
    \|\mathcal{C}(t+\epsilon,t)\|_1
= &
\operatorname{Tr}|\mathcal{C}_{\epsilon}^++\mathcal{C}_{\epsilon}^-|=
\operatorname{Tr}(\mathcal{C}_{\epsilon}^+)
+
\operatorname{Tr}(\mathcal{C}_{\epsilon}^-).
\end{align*}
Since the intermediate map is trace-preserving,
\begin{equation*}
    \operatorname{Tr}[\mathcal{C}(t+\epsilon,t)]=1 \implies \operatorname{Tr}[\mathcal{C}_{\epsilon}^+]
=
1+\operatorname{Tr}[\mathcal{C}_{\epsilon}^-]. 
\end{equation*}
So, $ \|\mathcal{C}(t+\epsilon,t)\|_1=1+2 \operatorname{Tr}[\mathcal{C}_{\epsilon}^-]. $
From Eq.~\eqref{RHP},
\begin{align} 
    g(t)
=
\lim_{\epsilon\to0^+}
\frac{
2\operatorname{Tr}(\mathcal{C}_{\epsilon}^-)
}{\epsilon}.   \label{RHPdeduction}
\end{align}  

Now, let $\mathcal W^+ \in \mathfrak{W}^+$ be an arbitrary positive witness with
largest and smallest eigenvalues $\lambda_+$ and $\lambda_-$.
Since
\begin{align*}
    \lambda_- \mathbb{I}
\le
\mathcal W^+
\le
\lambda_+ \mathbb{I},
\end{align*}
\begin{align}
   \operatorname{Tr}(\mathcal{C}^+_{\epsilon}\mathcal W^+)
\ge
\lambda_-\operatorname{Tr}(\mathcal{C}_{\epsilon}^+), \label{firstbound}
\end{align}  

and
\begin{align}
    \operatorname{Tr}(\mathcal{C}_{\epsilon}^-\mathcal W^+)
\le
\lambda_+\operatorname{Tr}(\mathcal{C}_{\epsilon}^-). \label{secondbound}
\end{align}

Combining Eqs.~\eqref{firstbound} and \eqref{secondbound},
\begin{equation}
\begin{split}
    \operatorname{Tr}[\mathcal{C}(t+\epsilon,t)\mathcal W^+]
= &
\operatorname{Tr}(\mathcal{C}_{\epsilon}^+\mathcal W^+)
-
\operatorname{Tr}(\mathcal{C}_{\epsilon}^-\mathcal W^+)
\\ & \ge
\lambda_-
+
(\lambda_--\lambda_+)\operatorname{Tr}(\mathcal{C}_{\epsilon}^-).
\end{split}
\end{equation}

Multiplying by $\frac{2}{\lambda_+-\lambda_-}$ gives,

\begin{align*}
\operatorname{Tr}[\mathcal{C}(t+\epsilon,t)\mathcal W_c^+]
\ge
\frac{2\lambda_-}{\lambda_+-\lambda_-}
-
2\operatorname{Tr}(\mathcal C_\epsilon^-).
\end{align*}

\begin{align*}
\begin{split}
\implies 
  -\operatorname{Tr}[\mathcal{C}(t+\epsilon,t)\mathcal W_c^+]
& \le
2\operatorname{Tr}(\mathcal{C}_{\epsilon}^-)
-
\frac{2\lambda_-}{\lambda_+-\lambda_-}\\ & 
\le
2\operatorname{Tr}(\mathcal{C}_{\epsilon}^-).
\end{split}
\end{align*}

Dividing by $\epsilon$ and taking the limit
$\epsilon\to0^+$

gives
\begin{align}\label{eq:wc_g}
  \widetilde w_c^+
\le
g(t).  
\end{align}
Since Eq.~\eqref{eq:wc_g} holds for every positive witness, we next derive the conditions under which this upper bound can be reached. 

For each fixed $\epsilon>0$, the negative contribution of the intermediate Choi operator is supported on  $\operatorname{supp}(\mathcal{C}_{\epsilon}^-)$.  The optimal positive witness may therefore depend on the intermediate Choi operator, and hence on $\epsilon$. So, we now allow the witness to form an $\epsilon$-adapted family $\mathcal{W}^{+}(\epsilon)$. Hence, for the $\epsilon$-adapted family, we choose

\begin{align*}
    \mathcal W^+=\mathcal W^+(\epsilon)=P_-(\epsilon),
\end{align*}
where $P_-(\epsilon)$ denotes the projector onto
$\operatorname{supp}(\mathcal{C}_{\epsilon}^-)$. Then for the projector, 
\begin{align*}
    \lambda_+=1,
\qquad
\lambda_-=0,
\end{align*}
and using Definition~\ref{def1},
\begin{align*}
    \mathcal W_c^+(\epsilon)=2P_-(\epsilon).
\end{align*}
Since $\mathcal W^+$ and hence $P_-(\epsilon)$ is supported on $\operatorname{supp}(\mathcal{C}_{\epsilon}^-)$, which is orthogonal to $\operatorname{supp}(\mathcal{C}_{\epsilon}^+)$,
\begin{equation*}
\mathcal{C}_{\epsilon}^+P_-(\epsilon)=0,
\qquad
\mathcal{C}_{\epsilon}^-P_-(\epsilon)=\mathcal{C}_{\epsilon}^-. 
\end{equation*}
Now, for each fixed $\epsilon>0$,
\begin{align*}
\operatorname{Tr}\;\big(\mathcal{C}(t+\epsilon,t)\,\mathcal W_{c}^+(\epsilon)\big)
&=
2\operatorname{Tr}\;\big(\mathcal{C}(t+\epsilon,t)\,P_-(\epsilon)\big)\\
&=
-2\operatorname{Tr}[\mathcal{C}_{\epsilon}^-].
\end{align*}
Dividing by $\epsilon$ and taking $\epsilon\to0^+$ on both sides give
\begin{align*}
    \lim_{\epsilon\to0^+}
    \frac{-\operatorname{Tr}\!\big(\mathcal{C}(t+\epsilon,t)\,\mathcal W_{c}^{+}(\epsilon)\big)}{\epsilon}
=
\lim_{\epsilon\to0^+}
\frac{2\operatorname{Tr}[\mathcal{C}_{\epsilon}^-]}{\epsilon}, 
\end{align*}
which is the RHP measure as given by Eq.~\eqref{RHPdeduction}. 
Thus, the upper bound in Eq.~\eqref{eq:wc_g} is saturated by the $\epsilon$-adapted family
$\{P_-(\epsilon)\}_{\epsilon>0}$. Hence, the RHP measure is generally obtained upon optimizing over all $\epsilon$-adapted families of positive witnesses, i.e., 
\begin{align}\label{eq:g_bound}
g(t)
=
\sup_{\mathcal W^+(\epsilon)\in\mathfrak W^+}
\left[
\lim_{\epsilon\to0^+}
\frac{
-\operatorname{Tr}\!\left(
\mathcal C(t+\epsilon,t)\mathcal W_c^+(\epsilon)
\right)
}{\epsilon}
\right].
\end{align}
If in addition, for any two distinct and sufficiently small $\epsilon,\epsilon'>0$, $\operatorname{supp}(\mathcal{C}_\epsilon^-)=\operatorname{supp}(\mathcal{C}_{\epsilon'}^-)$ (as verified for the amplitude-damping, thermal and dephasing channels in Section~\ref{Sec:Parity_Section}), then there exists a fixed projector $P_-$ such that $P_-(\epsilon)=P_-$ for all sufficiently small $\epsilon$. Consequently,
the $\epsilon$-adapted optimal witness reduces to a fixed positive
witness, and
\begin{align*}
    \widetilde w_c^+
    =
    g(t)
    =
    \max_{\mathcal W^+ \in \mathfrak{W}^+}
    \widetilde w_c^+,
\end{align*}
which establishes the optimality of $P_-$.

Since every positive witness is Hermitian, therefore, 
\begin{align*}
    \sup_{\mathcal{W}^{+} \in \mathfrak{W}^+}
    \widetilde w_c^{+}
    \leq
    \sup_{\mathcal{W}\in\mathfrak{W}_{\mathrm H}}\widetilde w_c.
\end{align*}
Using Eq.~\eqref{eq:g_bound}, this gives
\begin{align}g(t)
    \leq
    \sup_{\mathcal W\in\mathfrak{W}_{\mathrm H}}
    \widetilde w_c.
    \label{firstrelation}
\end{align}
Moreover, from Theorem~\ref{theorem1},
\begin{align}
    \sup_{\mathcal{W}\in\mathfrak{W}_{\mathrm{H}}} \widetilde{w}_c
    \leq
    D_T(t). \label{secondrelation}
\end{align}
From Eqs.~\eqref{firstrelation} and \eqref{secondrelation}, it follows that
\begin{align}
    g(t)
    \leq
    \sup_{\mathcal{W}\in\mathfrak{W}_{\mathrm{H}}} \widetilde{w}_c
    \leq
    D_T(t).
\end{align}

\end{proof}
Thus, Theorem \ref{thm:optimal} establishes a unified operational relation between the RHP, witness-based and distance-based measures of non-Markovianity.
\begin{remark}
While Theorem~\ref{theorem1} and Proposition~\ref{prop:tight} are formulated for general Hermitian witnesses, i.e., $\mathcal{W} \in \mathfrak{W}_{\rm H}$, Theorem~\ref{thm:optimal} specializes to the subclass of positive witnesses, $\mathcal{W}^+ \in \mathfrak{W}^+$. Nevertheless, the connection between the RHP, witness-based and distance-based measures remains valid for arbitrary Hermitian witnesses. Furthermore, since the distance-based measure is defined as the minimum distance between a non-Markovian Choi operator and the set of Markovian Choi operators, we explicitly distinguish these operators by representing them as $\mathcal{C}^{\mathrm{N}}$ and $\mathcal{C}^{\mathrm{M}}$, respectively, throughout Theorems~\ref{theorem1} and~\ref{thm:optimal}. In contrast, the witness-based quantity is evaluated as the expectation value of a witness over a single arbitrary Choi operator, without requiring any prior distinction between Markovian and non-Markovian Choi operators. Therefore, for notational simplicity, the superscripts "$\mathrm{N}$" and "$\mathrm{M}$" are omitted in the subsequent discussions. 
\end{remark} 

In the following section, we demonstrate the implementation of the proposed witnesses through parity measurements by providing illustrative examples.
\section{Parity witnesses as operationally optimal detectors of non-Markovianity}\label{Sec:Parity_Section}

We illustrate the proposed witness formalism using three paradigmatic models of open quantum system dynamics: (i) a non-Markovian amplitude-damping channel generated by a Lorentzian vacuum reservoir, (ii) a generalized amplitude-damping channel describing interaction with a finite-temperature bosonic reservoir, and (iii) the general Pauli channel, with pure dephasing as an important microscopic realization. These examples encompass the principal decoherence mechanisms of a qubit, namely energy relaxation, thermal excitation, and Pauli-type noise, respectively.

For each model, we derive the time-local master equation from the underlying system--environment interaction and construct the infinitesimal Choi operator associated with the intermediate dynamical map, i.e., 
\begin{equation}
\Lambda_{t+\epsilon,t}
\simeq
\mathbb{I}+\epsilon\mathcal L_t,
\qquad
\epsilon\ll1,
\end{equation}
where $\mathcal L_t$ denotes the time-local Lindblad generator. The negative eigenspace of the corresponding Choi operator determines the optimal witness according to Theorem~\ref{thm:optimal}. In Sec.~\ref{sec:operational_reconstruction}, we further derive experimentally accessible witness operators and show that they faithfully reproduce the RHP measure of non-Markovianity.

The corresponding Choi operators for the amplitude-damping channel, generalized amplitude-damping (GAD) channel and the Pauli (P) channel are denoted by $\mathcal{C}_{\rm Amp}(t+\epsilon,t)$, $\mathcal{C}_{\rm Th}(t+\epsilon,t)$ and $\mathcal{C}_{\rm P}(t+\epsilon,t)$ respectively, and their definition follows in accordance with Eq.~\eqref{choi}.  For convenience, we introduce the shorthand notation
$
\mathcal{C}_{\rm Amp}(t+\epsilon,t)
\equiv
\mathcal{C}_{\rm Amp}^{\epsilon}(t), \qquad
\mathcal{C}_{\rm Th}(t+\epsilon,t)
\equiv
\mathcal{C}_{\rm Th}^{\epsilon}(t), \qquad 
\mathcal{C}_{\rm P}(t+\epsilon,t)
\equiv
\mathcal{C}_{\rm P}^{\epsilon}(t).
$

We define the parity projectors,
\begin{align}
P_k^\pm
&=
\frac12
\left(
\mathbb{I}
\pm
\sigma_k\otimes\sigma_k
\right),
\qquad
k=x,y,z,
\label{eq:parity_projectors}
\end{align}
which project onto the $\pm1$ eigenspaces of the correlation observables $\sigma_k\otimes\sigma_k$.
Below, we provide illustrative examples.
\subsection{Examples of witness-based quantification of different non-markovian dynamics}\label{Sec:example}

\begin{example}\label{example1}
\textbf{Amplitude-damping channel:}
The time-local master equation governs the reduced dynamics \cite{Bellomo2007},
\begin{equation}
\label{eq:amp_master}
\frac{d\rho}{dt}
=
\mathcal L_t(\rho)
=
\gamma(t)
\left(
\sigma_-\rho\sigma_+
-
\frac12
\{\sigma_+\sigma_-,\rho\}
\right),
\end{equation}
where
$
\sigma_+=|1\rangle\langle0|,
\;
\sigma_-=|0\rangle\langle1|.
$

This dynamics arises from a two-level atom interacting with a bosonic reservoir described by the Hamiltonian
$H = H_S + H_B + H_I$, where $H_S = \Omega_0 \sigma_+ \sigma_-$, $H_B = \sum_k \Omega_k b_k^\dagger b_k$, and $H_I = \sum_k \left( g_k \sigma_+ b_k + g_k^* \sigma_- b_k^\dagger \right)$.

For a Lorentzian reservoir spectral density
\begin{align}
J(\Omega)
=
\frac{\gamma_0\lambda^2}
{2\pi
\left[
(\Omega-\Omega_0)^2+\lambda^2
\right]},
\label{eq:Lorentzian_SD}
\end{align}
the decay rate is
\begin{align}
\gamma(t)
=
-2
\operatorname{Re}
\left[
\frac{\dot G(t)}{G(t)}
\right],
\end{align}
where
\begin{align}
G(t)
=
e^{-\lambda t/2}
\left[
\cosh\!\left(\frac{t\Delta }{2}\right)
+
\frac{\lambda}{\Delta}
\sinh\!\left(\frac{t\Delta }{2}\right)
\right],
\end{align}
with
\begin{align}
\Delta
=
\sqrt{\lambda^2-2\gamma_0\lambda}.
\end{align}
The dynamics is CP divisible for $\lambda>2\gamma_0$, while $\lambda<2\gamma_0$ leads to temporarily negative decay rates and hence CP-indivisible non-Markovian evolution \cite{Bellomo2007, mukherjee2015efficiency}.

The corresponding infinitesimal Choi matrix is
\begin{equation}
\label{choi_amp}
\mathcal{C}_{\rm Amp}^{\epsilon}(t)=
\frac12
\begin{pmatrix}
1
&
0
&
0
&
1-\frac{\gamma(t)\epsilon}{2}
\\
0
&
0
&
0
&
0
\\
0
&
0
&
\gamma(t)\epsilon
&
0
\\
1-\frac{\gamma(t)\epsilon}{2}
&
0
&
0
&
1-\gamma(t)\epsilon
\end{pmatrix}.
\end{equation}

Its eigenvalues are
\begin{align}
\label{eq:eigen_amp}
\lambda_1
&=
0,
\qquad
\lambda_2
=
\frac{\gamma(t)\epsilon}{2},
\nonumber\\
\lambda_3
&=
\frac14
\left[
2-\gamma(t)\epsilon
-
\sqrt{
4
+
2\gamma(t)\epsilon
\left(
-2+\gamma(t)\epsilon
\right)
}
\right],
\nonumber\\
\lambda_4
&=
\frac14
\left[
2-\gamma(t)\epsilon
+
\sqrt{
4
+
2\gamma(t)\epsilon
\left(
-2+\gamma(t)\epsilon
\right)
}
\right].
\end{align}

From Eq.~\eqref{eq:eigen_amp}, the only eigenvalue that can become negative is
\begin{align}
\lambda_2
=
\frac{\gamma(t)\epsilon}{2},
\end{align}
whose eigenspace is the odd-parity $z$ sector,
\begin{align}
P_-
=
P_z^-.
\end{align}

Hence, by Theorem~\ref{thm:optimal},
\begin{align}
W_{\rm Amp}^{\rm opt}
=
P_z^-
=
\frac12
\left(
\mathbb{I}-\sigma_z\otimes\sigma_z
\right).
\end{align}

Using Definition~\ref{def1},
\begin{align}
W_c^{\rm Amp}
=
2W_{\rm Amp}^{\rm opt}
=
\mathbb{I}-\sigma_z\otimes\sigma_z,
\end{align}
and therefore
\begin{align}
\label{w_c_amp}
\widetilde w_c^{\rm Amp}
=
-
\frac{
\Tr
\left(
\mathcal{C}_{\rm Amp}^{\epsilon}(t)
W_c^{\rm Amp}
\right)
}{\epsilon}
=
-\gamma(t).
\end{align}
\end{example}

\vspace{0.5cm}

\begin{example}\label{example2}
\textbf{Thermal channel:}
The time-local master equation governs the reduced dynamics~(see Sec.~3.4.5 of Ref.~\cite{breuer}),
\begin{equation}
\label{eq:Th_master}
\begin{aligned}
\frac{d\rho}{dt}
&=
\mathcal L_t(\rho)
=
\gamma_\downarrow(t)
\left(
\sigma_-\rho\sigma_+
-
\frac12
\{\sigma_+\sigma_-,\rho\}
\right)
\\
&\quad+
\gamma_\uparrow(t)
\left(
\sigma_+\rho\sigma_-
-
\frac12
\{\sigma_-\sigma_+,\rho\}
\right),
\end{aligned}
\end{equation}
where
\begin{align}
\gamma_\downarrow(t)
&=
\gamma(t)
\left[
N(\Omega_0)+1
\right],
\nonumber\\
\gamma_\uparrow(t)
&=
\gamma(t)
N(\Omega_0),
\end{align}
with
\begin{align}
N(\Omega_0)
=
\frac{1}
{e^{\beta\Omega_0}-1},
\qquad
\beta=\frac{1}{k_BT},
\end{align}
being the mean thermal occupation number.

This dynamics arises from a two-level atom interacting with a bosonic reservoir initially prepared in thermal equilibrium. The total Hamiltonian is
$
H=H_S+H_B+H_I,
$
where
$
H_S=\Omega_0\sigma_+\sigma_-,
$
$
H_B=\sum_k\Omega_kb_k^\dagger b_k,
$
and
$
H_I=\sum_k
\left(
g_k\sigma_+b_k
+
g_k^*\sigma_-b_k^\dagger
\right).
$

The reservoir is initially prepared in the Gibbs state
\begin{align}
\rho_B
=
\frac{e^{-\beta H_B}}
{\Tr\!\left(e^{-\beta H_B}\right)}.
\end{align}

As in the vacuum amplitude-damping model, we consider the same Lorentzian reservoir spectral density given in Eq.~\eqref{eq:Lorentzian_SD}. Consequently, the time-dependent decay rate $\gamma(t)$ remains identical to that of the vacuum amplitude-damping model, while the finite temperature modifies only the emission and absorption rates through the thermal occupation number $N(\Omega_0)$.

The corresponding infinitesimal Choi matrix is
\begin{equation}
\label{eq:Th_choi}
\mathcal{C}_{\rm Th}^{\epsilon}(t)
=
\frac12
\begin{pmatrix}
1-\gamma_\uparrow\epsilon
&
0
&
0
&
1-\frac{(\gamma_\uparrow+\gamma_\downarrow)\epsilon}{2}
\\
0
&
\gamma_\uparrow\epsilon
&
0
&
0
\\
0
&
0
&
\gamma_\downarrow\epsilon
&
0
\\
1-\frac{(\gamma_\uparrow+\gamma_\downarrow)\epsilon}{2}
&
0
&
0
&
1-\gamma_\downarrow\epsilon
\end{pmatrix}.
\end{equation}

Its eigenvalues are
\begin{align}
\lambda_1 &=
\frac{\gamma_{\downarrow}(t)\epsilon}{2}, \nonumber\\
\lambda_2 &=
\frac{\gamma_{\uparrow}(t)\epsilon}{2}, \nonumber\\
\lambda_3 &=
\frac{1}{4}\!\left[
2-(\gamma_{\downarrow}(t)+\gamma_{\uparrow}(t))\epsilon
-\sqrt{\delta}
\right], \nonumber\\
\lambda_4 &=
\frac{1}{4}\!\left[
2-(\gamma_{\downarrow}(t)+\gamma_{\uparrow}(t))\epsilon
+\sqrt{\delta}
\right], 
\end{align}
where \begin{equation*}
\delta=
4-4\gamma_{\downarrow}(t)\epsilon+
\left[
2\gamma_{\downarrow}^2(t)
-2\gamma_{\downarrow}(t)\gamma_{\uparrow}(t)
+\gamma_{\uparrow}^2(t)
\right]\epsilon^2 .
\label{Delta}
\end{equation*}

From the above expression, the only eigenvalues that can become negative are
\begin{align}
\lambda_1
=
\frac{\gamma_\downarrow(t)\epsilon}{2},
\qquad
\lambda_2
=
\frac{\gamma_\uparrow(t)\epsilon}{2}.
\end{align}

The corresponding eigenspace coincides with the odd-parity $z$ sector,
\begin{align}
P_-
=
P_z^-.
\end{align}

Hence, by Theorem~\ref{thm:optimal},
\begin{align}
W_{\rm Th}^{\rm opt}
=
P_z^-
=
\frac12
\left(
\mathbb{I}-\sigma_z\otimes\sigma_z
\right).
\end{align}

Since $W_{\rm Th}^{\rm opt}$ is a projector, Definition~\ref{def1} gives
\begin{align}
W_c^{\rm Th}
=
2W_{\rm Th}^{\rm opt}
=
\mathbb{I}-\sigma_z\otimes\sigma_z.
\end{align}

Therefore,
\begin{align}
\label{eq:Th_witness}
\widetilde w_c^{\rm Th}
&=
-
\frac{
\Tr
\left(
\mathcal{C}_{\rm Th}^{\epsilon}(t)
W_c^{\rm Th}
\right)
}{\epsilon}
\nonumber\\
&=
-
\left(
\gamma_\uparrow(t)
+
\gamma_\downarrow(t)
\right)
\nonumber\\
&=
-
\left(
2N(\Omega_0)+1
\right)
\gamma(t).
\end{align}

In the zero-temperature limit,
$
N(\Omega_0)\rightarrow0,
$
the thermal channel continuously reduces to the vacuum amplitude-damping channel,
\begin{align}
\widetilde w_c^{\rm Th}
=
-\gamma(t),
\end{align}
which coincides with Eq.~\eqref{w_c_amp}.
\end{example}
\begin{example}\label{example3}
    
The dynamics of a qubit undergoing Pauli noise is described by the following time-local master equation~\cite{Hall2014},

\begin{equation}
\label{eq:pauli_master}
\begin{aligned}
\frac{d\rho}{dt}
&=
\mathcal{L}(\rho)
=
\gamma_x(t)(\sigma_x\rho\sigma_x-\rho)
+
\gamma_y(t)(\sigma_y\rho\sigma_y-\rho)
\\
&\quad
+
\gamma_z(t)(\sigma_z\rho\sigma_z-\rho).
\end{aligned}
\end{equation}

The corresponding Choi matrix is,

\begin{equation}
\mathcal{C}_{\rm P}^{\epsilon}(t)
=
\frac12
\begin{pmatrix}
a & 0 & 0 & b\\
0 & c & d & 0\\
0 & d & c & 0\\
b & 0 & 0 & a
\end{pmatrix},
\label{eq:Pauli_choi}
\end{equation}
where $a=1-\epsilon(\gamma_x+\gamma_y)$,
$b=1-\epsilon(\gamma_x+\gamma_y+2\gamma_z)$,
$c=\epsilon(\gamma_x+\gamma_y)$, and
$d=\epsilon(\gamma_x-\gamma_y)$.

The eigenvalues of this Pauli Choi~\eqref{eq:Pauli_choi},
\begin{equation}
\label{eq:eigen_pauli}
\begin{aligned}
\lambda_1
&=
1-\epsilon\bigl(\gamma_x(t)+\gamma_y(t)+\gamma_z(t)\bigr),
\qquad
\lambda_2
=
\epsilon\gamma_x(t),
\\
\lambda_3
&=
\epsilon\gamma_y(t),
\qquad
\lambda_4
=
\epsilon\gamma_z(t).
\end{aligned}
\end{equation}
Equation~(\ref{eq:eigen_pauli}) shows that the negative
subspace depends on which decay rates become negative.

Theorem~\ref{thm:optimal} therefore gives the adaptive
optimal witness
\begin{align}
W_{\rm Pauli}^{\rm opt}
=
\sum_{\gamma_i(t)<0}
|B_i\rangle\langle B_i|,
\end{align}

where
\begin{align}
|B_x\rangle
&=
|\Psi^+\rangle,
&
|B_y\rangle
&=
|\Psi^-\rangle,
&
|B_z\rangle
&=
|\Phi^-\rangle.
\end{align}

Since $W_{\rm Pauli}^{\rm opt}$ is a projector,
Definition~\ref{def1} gives
\begin{align}
W_c^{\rm Pauli}
=
2W_{\rm Pauli}^{\rm opt}
=
2
\sum_{\gamma_i(t)<0}
|B_i\rangle\langle B_i|.
\end{align}

Although the adaptive witness saturates the RHP measure,
it requires selecting Bell projectors according to the
signs of the instantaneous decay rates.

Motivated by experimental simplicity, we instead employ
the fixed parity witnesses
$
P_z^-,
P_x^-,
P_y^-.
$

The corresponding normalized witnesses are
\begin{align}
W_{c,z}
&=
\mathbb{I}-\sigma_z\otimes\sigma_z,
\\
W_{c,x}
&=
\mathbb{I}-\sigma_x\otimes\sigma_x,
\\
W_{c,y}
&=
\mathbb{I}-\sigma_y\otimes\sigma_y.
\end{align}

The witnessed quantities become
\begin{align}\label{w_c_Pauli}
\widetilde w_z
&=
-\frac{
\Tr\!\left(\mathcal{C}_{\rm P}^{\epsilon}(t)W_{c,z}\right)
}{\epsilon}
=
-2\bigl(\gamma_x(t)+\gamma_y(t)\bigr),
\nonumber\\
\widetilde w_x
&=
-\frac{
\Tr\!\left(\mathcal{C}_{\rm P}^{\epsilon}(t)W_{c,x}\right)
}{\epsilon}
=
-2\bigl(\gamma_y(t)+\gamma_z(t)\bigr),
\nonumber\\
\widetilde w_y
&=
-\frac{
\Tr\!\left(\mathcal{C}_{\rm P}^{\epsilon}(t)W_{c,y}\right)
}{\epsilon}
=
-2\bigl(\gamma_x(t)+\gamma_z(t)\bigr).
\end{align}

The three rates are reconstructed as
\begin{align}\label{Pauli_reconstruction}
\gamma_x(t)
&=
\frac{
\widetilde w_x-\widetilde w_y-\widetilde w_z
}{4},
\nonumber\\
\gamma_y(t)
&=
\frac{
\widetilde w_y-\widetilde w_x-\widetilde w_z
}{4},
\nonumber\\
\gamma_z(t)
&=
\frac{
\widetilde w_z-\widetilde w_x-\widetilde w_y
}{4}.
\end{align}

Consequently,
\begin{align}
g(t)
=
2\sum_{i=x,y,z}
\max\{0,-\gamma_i(t)\}.
\end{align}
\end{example}
The examples above show that Bell-parity witnesses provide
an operational realization of the optimal positive witness
construction of Theorem~\ref{thm:optimal}. For amplitude
damping and thermal channels, the negative part of the
intermediate Choi operator always belongs to a fixed
Bell-parity sector, yielding a universal parity witness
that exactly reproduces the RHP non-Markovianity rate.
For general Pauli dynamics, the optimal witness depends on
the instantaneous sign pattern of the decay rates, although
the full non-Markovianity can still be reconstructed from
parity measurements in the three Pauli bases.

To illustrate how the general Pauli reconstruction applies to a physically relevant microscopic model, we consider pure dephasing induced by a bosonic environment.

\subsubsection{Microscopic realization: Pure dephasing channel}\label{dephasing}

As a concrete realization of the Pauli channel, we consider a qubit coupled to a bosonic reservoir through a longitudinal interaction. In this case, the coupling preserves the qubit's energy and affects only its phase, giving rise to pure dephasing. Because the model is exactly solvable, it serves as a standard microscopic description of non-Markovian dephasing~\cite{Haikka2013} in many experimental platforms, such as solid-state qubits, superconducting circuits, and quantum optical systems. The total Hamiltonian is $H = H_S + H_B + H_I$, where
$H_S = \frac{\Omega_0}{2}\sigma_z$,
$H_B = \sum_k \Omega_k b_k^\dagger b_k$, and
$H_I = \sigma_z \sum_k \left( g_k b_k^\dagger + g_k^* b_k \right)$.

The Ohmic-class spectral density characterizes the reservoir
\begin{align}
J(\Omega)
=
\eta
\Omega_c^{1-s}
\Omega^s
e^{-\Omega/\Omega_c},
\label{eq:ohmic_sd}
\end{align}
where $\eta$ is the coupling strength, $\Omega_c$ is the cutoff frequency, and $s$ determines the reservoir class. The cases $s<1$, $s=1$, and $s>1$ correspond to sub-Ohmic, Ohmic, and super-Ohmic reservoirs, respectively.

The reduced dynamics is governed by the time-local master equation
\begin{align}
\label{eq:dephasing_master}
\frac{d\rho}{dt}
=
\gamma(t)
\left(
\sigma_z\rho\sigma_z-\rho
\right),
\end{align}
where
\begin{align}
\gamma(t)
=
\frac{d}{dt}\mathcal{F}(t),
\end{align}
with
\begin{align}
\mathcal{F}(t)
=
\frac12
\int_0^\infty
d\Omega\,
J(\Omega)
\frac{1-\cos(\Omega t)}
{\Omega^2}
\coth
\left(
\frac{\beta\Omega}{2}
\right).
\end{align}
Equivalently,
\begin{align}
\gamma(t)
=
\frac12
\int_0^\infty
d\Omega\,
J(\Omega)
\coth
\left(
\frac{\beta\Omega}{2}
\right)
\frac{\sin(\Omega t)}{\Omega}.
\end{align}

For simplicity, we restrict our analysis to the zero-temperature limit ($\beta\rightarrow\infty$), for which
\begin{align}
\coth\left(\frac{\beta\Omega}{2}\right)\rightarrow1.
\end{align}
Consequently,
\begin{align}
\mathcal{F}(t)
&=
\frac12
\int_0^\infty
d\Omega\,
J(\Omega)
\frac{1-\cos(\Omega t)}{\Omega^2},
\nonumber\\
\gamma(t)
&=
\frac12
\int_0^\infty
d\Omega\,
J(\Omega)
\frac{\sin(\Omega t)}{\Omega}.
\end{align}

Using the Laplace-type integral
$
\int_0^\infty x^{\nu-1}e^{-zx}\,dx
=
\Gamma(\nu)z^{-\nu}
$
with $z=\Omega_c^{-1}-it$ and $\Re(z)>0$, we obtain
\begin{align}
\gamma(t)
&=
\frac{\eta\,\Omega_c\,\Gamma(s)}{2}
\frac{\sin\!\left[s\arctan(\Omega_ct)\right]}
{\left(1+\Omega_c^2t^2\right)^{s/2}}.
\label{eq:gamma_t}
\end{align}
For the Ohmic case ($s=1$), $\gamma(t)=\eta\Omega_c^2t/[2(1+\Omega_c^2t^2)].$ At zero temperature, $\gamma(t)\ge0$ for $s\le2$, whereas for $s>2$ it vanishes at $t_k=\Omega_c^{-1}\tan(k\pi/s),$
where $k=1,2,\ldots,\lfloor(s-1)/2\rfloor$, and alternates in sign across successive zeros. Consequently, the intervals where $\gamma(t)<0$ coincide with the loss of complete positivity divisibility of the dynamical map and therefore characterize non-Markovian dynamics in the sense of the RHP criterion~\cite{RHP}. These intervals are directly detected by the witness introduced below.

Within the general Pauli-channel parametrization of Eq.~\eqref{eq:pauli_master}, the only nonvanishing decay rate is
\begin{align}
\gamma_x(t)
=
0,
\qquad
\gamma_y(t)
=
0,
\qquad
\gamma_z(t)
=
\gamma(t).
\end{align}

Substituting these rates into Eq.~\eqref{w_c_Pauli} immediately yields
\begin{align}
\widetilde w_z
&=
0,
\nonumber\\
\widetilde w_x
&=
-2\gamma(t),
\nonumber\\
\widetilde w_y
&=
-2\gamma(t).
\end{align}

thereby recovering the microscopic dephasing generator.

Therefore, whenever $\gamma(t)<0$, the corresponding witness rates $\widetilde w_x$ and $\widetilde w_y$ become positive, providing an operational signature of CP-indivisible dynamics. In the present convention, the corresponding non-Markovianity rate is
\begin{align}
g(t)
=
\max
\left\{
0,
-2\;\gamma(t)
\right\}.
\end{align}

Although the witness formalism is naturally expressed in
terms of the intermediate Choi operator
$
\mathcal{C}(t+\epsilon,t),
$
the corresponding intermediate propagator
$
\Lambda_{t+\epsilon,t}
=
\Lambda_{t+\epsilon,0}\Lambda_{t,0}^{-1}
$
need not be completely positive and therefore does not
represent a physical quantum channel. Nevertheless, the
witness values and the associated RHP
non-Markovianity rate can be reconstructed from
experimentally accessible two-qubit correlations measured
on a maximally entangled probe evolving under the physical
dynamics.

\subsection{Operational reconstruction of the intermediate Choi operator from correlation measurements}
\label{sec:operational_reconstruction}

The witness construction introduced above is formulated in terms of
the intermediate Choi operator associated with
$\Lambda_{t+\epsilon,t}$. In the non-Markovian regime, however, this
intermediate map need not be completely positive and therefore cannot,
in general, be implemented as a physical quantum channel \cite{RHPreview}. To make our proposed witness construction experimentally accessible, we therefore seek to extract the same witness information from the physical evolution itself, without directly implementing the intermediate map. The same information can instead be accessed through measurements performed along the physical evolution from the initial time to $t$.


The basic idea is to prepare a maximally entangled ancilla--system
pair, allow only the system qubit to interact with the environment,
and then measure suitable two-qubit correlations. The physical state
available at time $t$ is
\begin{equation}
\rho_{AS}(t)
=
\left(
\mathbb I_A\otimes\Lambda_{t,0}
\right)
\left(
|\Phi^+\rangle\langle\Phi^+|
\right),
\label{eq:physical_state}
\end{equation}
where
\begin{equation}
|\Phi^+\rangle_{AS}
=
\frac{|00\rangle+|11\rangle}{\sqrt2}.
\end{equation}

The complete procedure is illustrated in
Fig.~\ref{fig:operational_protocol}. The three shaded regions of the
circuit correspond to the state preparation, the physical
open-system evolution, and the final correlation measurement,
respectively.

\begin{figure*}[t]
    \centering
    \includegraphics[
        width=0.9\textwidth]{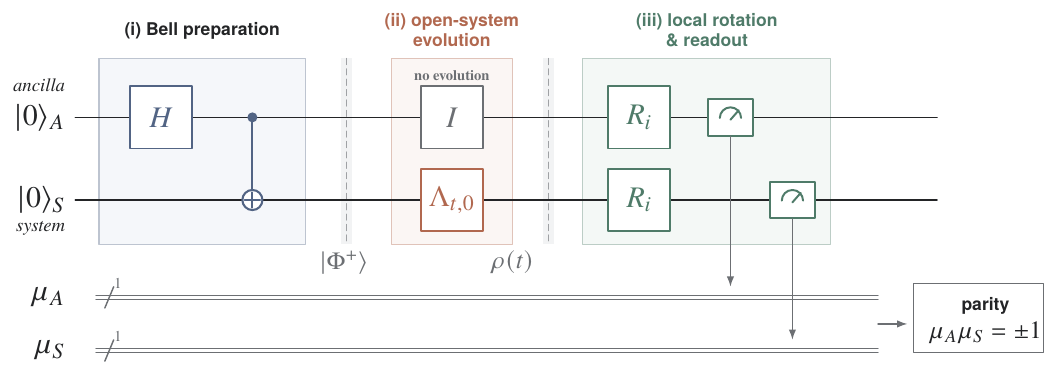}
    \caption{\footnotesize
\textbf{Single-run experimental circuit for measuring a two-qubit
correlation.}
\textbf{(i) Bell preparation:}
the ancilla $A$ and system $S$ are initialized in
$|0\rangle_A|0\rangle_S$, and a Hadamard gate followed by a CNOT
prepares
$|\Phi^+\rangle_{AS}
=(|00\rangle+|11\rangle)/\sqrt2$.
\textbf{(ii) Open-system evolution:}
only the system qubit evolves under the physical map
$\Lambda_{t,0}$, while the ancilla remains isolated, producing
$\rho_{AS}(t)
=(\mathbb I_A\otimes\Lambda_{t,0})
(|\Phi^+\rangle\langle\Phi^+|)$.
\textbf{(iii) Local rotation and readout:}
the correlation
$\mathcal S_i(t)=\langle\sigma_i\otimes\sigma_i\rangle_t$
is measured by applying the corresponding local basis rotation
$R_i$ to both qubits, with
$R_i^\dagger\sigma_zR_i=\sigma_i$, followed by computational-basis
readout. The measured outcomes are mapped to
$\mu_A,\mu_S=\pm1$, and their product gives the parity of that run.
Repeating the circuit many times at the same $t$ yields the
correlation $\mathcal S_i(t)$.}
    \label{fig:operational_protocol}
\end{figure*}

\textbf{\textit{Step 1: Bell-state preparation.}} The first shaded region of Fig.~\ref{fig:operational_protocol},
labelled \textbf{(i) Bell preparation}, prepares the ancilla-assisted
probe state. The two qubits are initialized in
$|0\rangle_A|0\rangle_S$. A Hadamard gate is applied to the ancilla,
followed by a CNOT with the ancilla as the control and the system as
the target. The state at the end of this region is therefore
\begin{equation}
|\Phi^+\rangle_{AS}
=
\frac{|00\rangle+|11\rangle}{\sqrt2}.
\label{eq:bell_probe}
\end{equation}

\textbf{\textit{Step 2: Physical open-system evolution.}} The second shaded region of Fig.~\ref{fig:operational_protocol},
labelled \textbf{(ii) open-system evolution}, contains the dynamics
to be characterized. The ancilla remains isolated, while the system
qubit evolves under the physical map $\Lambda_{t,0}$ for a
controllable time $t$. The Bell state is therefore transformed as
\begin{equation}
|\Phi^+\rangle\langle\Phi^+|
\longrightarrow
\rho_{AS}(t)
=
(\mathbb I_A\otimes\Lambda_{t,0})
(|\Phi^+\rangle\langle\Phi^+|).
\end{equation}

Importantly, the experimental protocol requires only the physical evolution $\Lambda_{t,0}$. The intermediate propagator
$\Lambda_{t+\epsilon,t}$, whose lack of complete positivity signals
non-Markovianity, is never required as an experimentally realizable
operation.

\textbf{\textit{Step 3: Correlation measurement.}} The third shaded region of Fig.~\ref{fig:operational_protocol},
labelled \textbf{(iii) local rotation \& readout}, is used to measure
a two-qubit correlation of the form
\begin{equation}
\mathcal{S}_i(t)
=
\left\langle
\sigma_i\otimes\sigma_i
\right\rangle_t,
\qquad
i\in\{x,y,z\}.
\label{eq:general_correlation}
\end{equation}
At this stage, the index $i$ is left unspecified. Its value is fixed
only after the relevant witness for a given dynamics has been
identified. In other words, the experimental architecture remains
the same; only the local measurement basis changes.

Since the final qubit readout is performed in the computational
basis, the desired Pauli basis is first rotated into the $Z$ basis.
This is the role of the two $R_i$ gates shown in the third region of
Fig.~\ref{fig:operational_protocol}. They are chosen such that
\begin{equation}
R_i^\dagger\sigma_zR_i
=
\sigma_i.
\label{eq:Ri_condition}
\end{equation}
For the three Pauli bases, one may choose
\begin{equation}
R_z=\mathbb I,
\qquad
R_x=H,
\qquad
R_y=HS^\dagger,
\label{eq:Ri_choices}
\end{equation}
where
$S=
\begin{pmatrix}
1 & 0\\
0 & i
\end{pmatrix}.$

Thus, $R_i$ should not be interpreted as an additional part of the
open-system dynamics. It simply specifies which correlation is being
read out. Once the witness analysis identifies the relevant Pauli
correlation, the corresponding $R_i$ is selected in the same circuit.
If several correlations are required, the experiment is repeated
with the corresponding measurement settings.

After the local rotations, both qubits are measured in the
computational basis, with outcomes
$b_A,b_S\in\{0,1\}$. We assign the corresponding eigenvalues
\begin{equation}
\mu_A=(-1)^{b_A},
\qquad
\mu_S=(-1)^{b_S},
\end{equation}
so that the single-shot parity is
\begin{equation}
\pi_i=\mu_A\mu_S=\pm1.
\end{equation}
Averaging over repeated runs gives
\begin{equation}
\mathcal{S}_i(t)
=
\langle\pi_i\rangle
=
\frac{
N_{00}+N_{11}-N_{01}-N_{10}
}{
N_{\rm tot}
},
\label{eq:correlation_counts}
\end{equation}
where
\[
N_{\rm tot}
=
N_{00}+N_{01}+N_{10}+N_{11},
\]
and $N_{b_A b_S}$ denotes the number of runs yielding the joint
outcome $(b_A,b_S)$. Thus, $N_{00}$ and $N_{11}$ correspond to
correlated outcomes, while $N_{01}$ and $N_{10}$ correspond to
anti-correlated outcomes in the chosen measurement basis.

Measurements of the Pauli-parity observables considered here are
compatible with several established experimental platforms:
in superconducting circuits, two-qubit parity and Pauli correlations
can be accessed through basis rotations and joint qubit
readout~\cite{Riste2013,Steffen2012}; in polarization-encoded photonic
systems, the corresponding correlations are measured using local
polarization rotations and coincidence-resolved
detection~\cite{James2001,Barbieri2003}; and in trapped-ion systems,
parity correlations are routinely extracted after suitable qubit
rotations followed by state-dependent fluorescence
measurements~\cite{Leibfried2003,Bruzewicz2019}.

\textit{Step 4: Estimation of the instantaneous correlation rate.} The first three steps constitute a single experimental run at a
chosen evolution time $t$. To obtain the time-local information
required by the witness, the complete circuit is repeated at nearby
times $t$ and $t+\delta t$. The derivative of the measured
correlation is then estimated as
\begin{equation}
\dot{\mathcal{S}}_i (t)
\simeq
\frac{
\mathcal{S}_i (t+\delta t)-\mathcal{S}_i (t)
}{
\delta t
}.
\label{eq:finite_difference}
\end{equation}

Since only the system qubit evolves, the correlation satisfies
\begin{equation}
\frac{d}{dt}\mathcal{S}_i(t)
=
{\rm Tr}
\left[
\rho_{AS}(t)
\left(
\sigma_i\otimes
\mathcal L_t^\dagger(\sigma_i)
\right)
\right],
\label{eq:heisenberg_correlation}
\end{equation}
where $\mathcal L_t^\dagger$ denotes the adjoint of the time-local
generator.

It is convenient to introduce the experimentally accessible
logarithmic correlation rate
\begin{equation}
\nu_i(t)
=
\frac{\dot{\mathcal S}_i(t)}{\mathcal S_i(t)}
=
\frac{d}{dt}\ln\!\left|\mathcal S_i(t)\right|,
\qquad
\mathcal S_i(t)\neq 0.
\label{eq:operational_rate}
\end{equation}
For the dynamical models considered below, this directly measurable
quantity is related to the corresponding witness rate
$\widetilde w_i(t)$, with the precise relation established separately
for each model.
The relation between $\widetilde w_i(t)$ and the decay rates depends
on the particular dynamical generator. We derive these relations
below for the channels introduced in Sec.~\ref{Sec:example}; this
will also determine which correlation, and therefore which local
rotation $R_i$, is required in each case. The corresponding
Heisenberg-picture derivations are given in
Appendix~\ref{app:heisenberg}.
\subsubsection{Amplitude damping}\label{sub_amp_damp}

For the amplitude-damping channel, define
\begin{align}
\mathcal{S}_z(t)
=
\langle
\sigma_z\otimes\sigma_z
\rangle_t.
\end{align}

Using the Heisenberg-picture result derived in
Appendix~\ref{app:heisenberg},
\begin{align}
\frac{d}{dt}
\mathcal{S}_z(t)
=
-\gamma(t)\mathcal{S}_z(t),
\end{align}
which immediately gives
\begin{align}
\gamma(t)
=
-
\frac{d}{dt}
\ln \left|\mathcal{S}_z(t)\right|.
\end{align}

Using Eq.~\eqref{w_c_amp},
\begin{align}
\widetilde w_c^{\rm Amp}
=
\frac{d}{dt}
\ln \left|\mathcal{S}_z(t)\right|,
\end{align}

and therefore
\begin{align}
g(t)
=
\max\{0,\widetilde w_c^{\rm Amp}\}.
\end{align}
\begin{figure}[t]
\centering
\includegraphics[width=0.92\columnwidth]{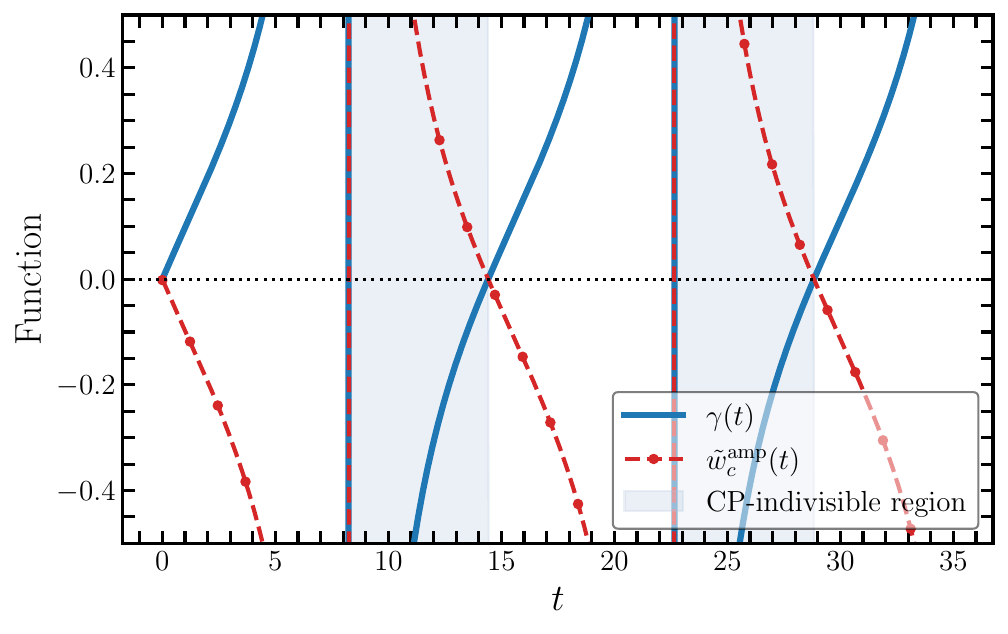}
\caption{\footnotesize
Operational reconstruction of the non-Markovianity witness for the
amplitude-damping channel. The reconstructed witness
$\widetilde{w}_{c}^{\rm Amp}(t)$ and the exact decay rate $\gamma(t)$
are plotted as functions of time for a Lorentzian reservoir with
$\lambda/\gamma_{0}=0.1$, corresponding to the non-Markovian regime
$\lambda<2\gamma_{0}$~\cite{Bellomo2007, mukherjee2015efficiency}. The shaded regions indicate the
CP-indivisible regime, where $\gamma(t)<0$.
}
\label{fig:amp_witness}
\end{figure}

\subsubsection{Thermal channel}
Similarly, Appendix~\ref{app:heisenberg} gives for the
generalized amplitude-damping channel
\begin{align}
\frac{d}{dt}
\mathcal{S}_z(t)
=
-(2N(\Omega_0)+1)\gamma(t)\mathcal{S}_z(t),
\end{align}
leading to
\begin{align}
\widetilde w_c^{\rm Th}
=
\frac{d}{dt}
\ln \left|\mathcal{S}_z(t)\right|,
\end{align}
and
\begin{align}
g(t)
=
\max\{0,\widetilde w_c^{\rm Th}\}.
\end{align}

\vspace{0.5cm}
Since $\mathcal{S}_z(t)>0$ along the Bell-state trajectory for the
amplitude-damping channel, negative decay rates are equivalent to
a positive logarithmic growth rate of the Bell correlation,
\begin{align}
\gamma(t)<0
\iff
\frac{d}{dt}\ln\!\left|\mathcal{S}_z(t)\right|>0
\iff
\widetilde w_c^{\rm Amp}(t)>0.
\end{align}
Figure~\ref{fig:amp_witness} shows the operational reconstruction
of the witness from the experimentally accessible Bell correlation.
The reconstructed witness faithfully reproduces the exact decay rate,
with every interval of CP-indivisible dynamics ($\gamma(t)<0$)
identified by a positive witness value.

For the thermal channel, the witness is simply rescaled as
\begin{align}
\widetilde w_c^{\rm Th}(t)
=
-(2N(\Omega_0)+1)\gamma(t).
\end{align}
Since $(2N(\Omega_0)+1)>0$, the sign of the witness remains unchanged,
so that
\begin{align}
\gamma(t)<0
\iff
\widetilde w_c^{\rm Th}(t)>0.
\end{align}
Therefore, the same operational criterion applies to the thermal channel, although we do not present a separate numerical illustration because it differs from the amplitude-damping case only by the positive scaling factor $(2N(\Omega_0)+1)$.
\subsubsection{Pauli channel}\label{sub_pauli}

For the Pauli dynamical map, define
\begin{align}
\mathcal{S}_i(t)
=
\langle
\sigma_i\otimes\sigma_i
\rangle_t,
\qquad
i=x,y,z.
\end{align}
Using the Heisenberg-picture equations derived in
Appendix~\ref{app:heisenberg},
\begin{align}
\frac{d}{dt}
\mathcal{S}_x(t)
&=
-2(\gamma_y+\gamma_z)\mathcal{S}_x(t),
\nonumber\\
\frac{d}{dt}
\mathcal{S}_y(t)
&=
-2(\gamma_x+\gamma_z)\mathcal{S}_y(t),
\nonumber\\
\frac{d}{dt}
\mathcal{S}_z(t)
&=
-2(\gamma_x+\gamma_y)\mathcal{S}_z(t).
\end{align}

Consequently,
\begin{align}
\widetilde w_x(t)
&=\nu_x(t)
=\frac{d}{dt}\ln\!\left|\mathcal S_x(t)\right|,
\nonumber\\
\widetilde w_y(t)
&=\nu_y(t)
=\frac{d}{dt}\ln\!\left|\mathcal S_y(t)\right|,
\nonumber\\
\widetilde w_z(t)
&=\nu_z(t)
=\frac{d}{dt}\ln\!\left|\mathcal S_z(t)\right|.
\end{align}

Whenever one of the combinations
$\gamma_y+\gamma_z$,
$\gamma_x+\gamma_z$, or
$\gamma_x+\gamma_y$
becomes negative, the corresponding Bell correlation
exhibits a revival, providing a direct operational
signature of information backflow.
\begin{figure}[t]
\centering
\includegraphics[width=0.92\columnwidth]{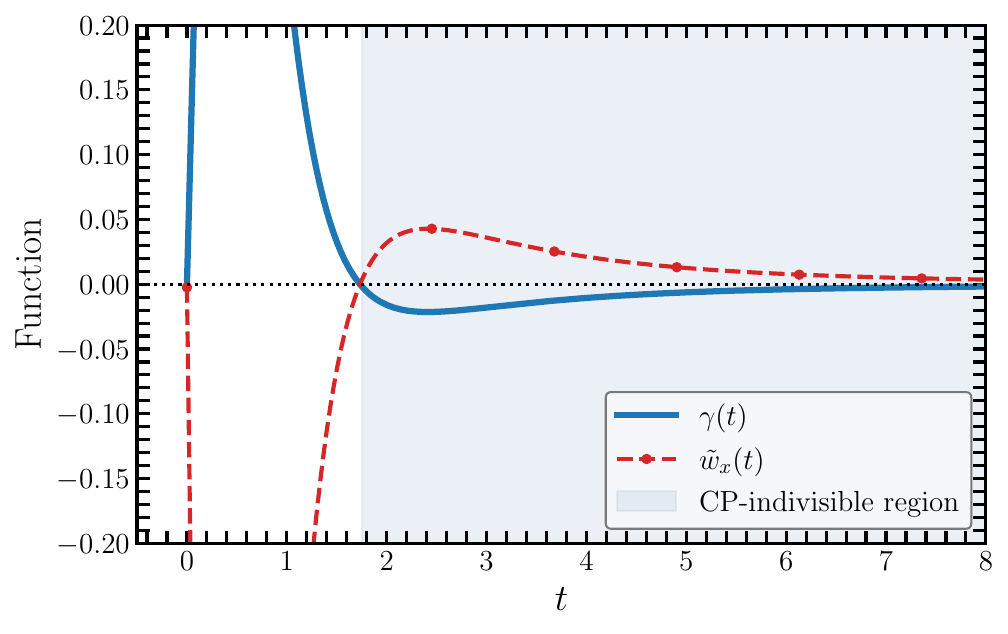}
\caption{\footnotesize
Operational reconstruction of the non-Markovianity witness for the
pure-dephasing channel. The reconstructed witness
$\widetilde{w}_{x}(t)$ and the exact decay rate $\gamma(t)$ are plotted
as functions of time for $\eta=1$, $\Omega_c=1$, and $s=3$, corresponding
to a super-Ohmic reservoir~\cite{Haikka2013}. The shaded regions indicate the
CP-indivisible regime, where $\gamma(t)<0$.
}
\label{fig:dephasing_witness}
\end{figure}

\begin{table*}[t]
\caption{\footnotesize
Number of independent correlations required at each sampled
evolution time for a complete reconstruction of the intermediate
Choi operator and for our witness-based method. Complete
reconstruction requires $15$ nontrivial Pauli expectation values,
whereas our method requires only the correlations listed in
the last column.
}
\label{tab:measurement_overhead}
\begin{ruledtabular}
\begin{tabular}{lccc}
Dynamics
&
Complete Choi reconstruction
&
Our method
&
Measured correlation(s)
\\
\hline
Amplitude damping
&
$15$
&
$1$
&
$\mathcal{S}_z(t)$
\\
Thermal
&
$15$
&
$1$
&
$\mathcal{S}_z(t)$
\\
General Pauli
&
$15$
&
$3$
&
$\mathcal{S}_x(t),\,\mathcal{S}_y(t),\,\mathcal{S}_z(t)$
\end{tabular}
\end{ruledtabular}
\end{table*}
The decay rates are reconstructed from the measured
witness values according to
\begin{align}
\gamma_x(t)
&=
\frac{
\widetilde w_x(t)
-
\widetilde w_y(t)
-
\widetilde w_z(t)
}{4},
\nonumber\\
\gamma_y(t)
&=
\frac{
\widetilde w_y(t)
-
\widetilde w_x(t)
-
\widetilde w_z(t)
}{4},
\nonumber\\
\gamma_z(t)
&=
\frac{
\widetilde w_z(t)
-
\widetilde w_x(t)
-
\widetilde w_y(t)
}{4},
\label{eq:gamma_reconstruction}
\end{align}
from which the RHP non-Markovianity rate follows as
\begin{align}
g(t)
=
2\sum_{i=x,y,z}
\max\{0,-\gamma_i(t)\}.
\label{eq:rhp_reconstructed}
\end{align}

As a physical example, in Sec.~\ref{dephasing} we consider the
microscopic pure-dephasing model. In this case the Pauli
generator of Eq.~\eqref{eq:pauli_master} reduces to
\begin{align}
\gamma_x(t)
=
0,
\qquad
\gamma_y(t)
=
0,
\qquad
\gamma_z(t)
=
\gamma(t),
\end{align}
where the explicit form of $\gamma(t)$ is determined by the
bosonic reservoir spectral density in Eq.~\eqref{eq:ohmic_sd}.

Substituting these rates into the general witness relations
gives
\begin{align}
\widetilde w_z
&=
0,
\nonumber\\
\widetilde w_x
&=
-2\gamma(t),
\nonumber\\
\widetilde w_y
&=
-2\gamma(t),
\end{align}

Consequently,
\begin{align}
g(t)
=
\max\{0,-2\gamma(t)\},
\end{align}
Thus, a negative decay rate ($\gamma(t)<0$) manifests operationally as a revival of the Bell correlations $\mathcal{S}_x(t)$ and $\mathcal{S}_y(t)$, or equivalently as a positive reconstructed witness $\widetilde w_x(t)$. Figure~\ref{fig:dephasing_witness} shows that the reconstructed witness accurately reproduces the exact dephasing rate, with positive witness values identifying all CP-indivisible intervals.

A direct evaluation of \(g(t)\) from Eq.~\eqref{RHP} requires reconstructing the intermediate Choi operator \(\mathcal{C}(t+\epsilon,t)\) at each sampled interval, which in turn involves measuring the \(15\) nontrivial Pauli expectation values of the corresponding two-qubit state. Our method avoids this complete reconstruction by extracting the required information from a small set of experimentally accessible two-qubit correlations. As summarized in Table~\ref{tab:measurement_overhead}, only one correlation is required for amplitude-damping and thermal dynamics, while three correlations are sufficient for general Pauli dynamics at each sampled time.

The same correlations are measured at neighboring evolution times, \(t\) and \(t+\delta t\), allowing their time derivatives to be estimated. These derivatives determine the corresponding time-local decay rates, from which the witness values and, ultimately, the instantaneous RHP non-Markovianity rate \(g(t)\) can be obtained.

\section{Conclusions}\label{Conclusion} 
Quantifying the memory content of an open quantum evolution is a
prerequisite for exploiting non-Markovianity as a resource in quantum
technologies \cite{task1,task2,task3,task4,task5}. However, both the RHP and geometric distance measures require complete knowledge of the intermediate dynamical map, making them experimentally challenging \cite{Bhattacharya20,RHPreview}. In this work, we have developed a witness-based framework that circumvents this requirement and enables direct certification of non-Markovianity from experimentally accessible expectation values.

Specifically, we first show that any non-Markovianity witness capable of detecting a given dynamics provides a certifiable lower bound on the geometric distance measure. We have further
characterized the exact conditions under which this bound is saturated. However, in consideration of the experimental realizability of this witness operator, we focus our study on positive semidefinite witnesses. After suitable normalization, such witnesses constitute valid POVM elements, making their expectation values directly accessible as measurement probabilities. We then proved that every positive witness certifies a lower bound on the RHP measure, with equality
attained by the projector onto the negative eigenspace of the intermediate
Choi operator. Taken together, these results yield the quantitative hierarchy among our proposed framework, RHP measure, and geometric distance measure.

Further to demonstrate the practical utility of our framework, we apply it to two paradigmatic qubit dynamical models: amplitude damping and general Pauli dynamics. For both models, the optimal positive witnesses reduce to simple Bell-parity projectors. We show that a single universal parity witness exactly recovers the RHP measure for amplitude damping, while three parity witnesses suffice for general Pauli dynamics. Importantly, these witnesses can be extracted directly from experimentally accessible Bell correlations, without constructing the intermediate Choi operator. Thus, the measurement cost is reduced to one correlation for amplitude damping and three for general Pauli dynamics, compared with fifteen parameters required for complete two-qubit Choi state reconstruction. The resulting witnesses faithfully reproduce the exact decay rates and identify all CP-indivisible intervals, demonstrating an efficient and operationally accessible approach to quantifying non-Markovianity.

Our analysis suggests several directions for further investigation. A natural
question is whether the parity structure of the optimal witnesses persists in
multiqubit dynamics \cite{PhysRevA.90.052104}, in which case the present scheme would offer an
exponentially larger saving over complete reconstruction of the Choi state. It would also be
worthwhile to investigate the robustness of the reconstruction against statistical noise and finite-time resolution, as well as the use of randomized-measurement techniques such as shadow tomography \cite{Huang2020,PhysRevA.109.022247} to estimate multiple parity witnesses simultaneously from the same experimental data. Finally, an important step toward practical implementation is to realize the present correlation-based reconstruction experimentally. Non-Markovian dynamics has already
been probed in superconducting quantum processors
\cite{White2020,Zhang2022,Gaikwad2024} and in several photonic
platforms \cite{Liu2011,Bernardes2015,Cialdi2019,Wu2020}, making
these systems promising settings for testing our scheme. In particular, it would be interesting to examine whether the RHP rate can be
recovered experimentally using only the reduced set of correlations
identified here, without reconstructing the full intermediate Choi
operator.

%

\appendix
\onecolumngrid
\section{Derivation of the Heisenberg-Picture Relations}
\label{app:heisenberg}

For completeness, we derive the Heisenberg-picture relations used to reconstruct the decay rates from experimentally accessible correlation functions.

Let
\begin{align}
\rho(t)
=
(\mathbb I\otimes\Lambda_{t,0})
\left(
|\Phi^+\rangle\langle\Phi^+|
\right),
\end{align}
where only the system qubit evolves while the ancilla remains isolated.

For an arbitrary system observable $O$, the expectation value is
\begin{align}
\langle
\mathbb I\otimes O
\rangle_t
=
{\rm Tr}
\left[
\rho(t)
(\mathbb I\otimes O)
\right].
\end{align}

Differentiating with respect to time and using
\begin{align}
\dot{\rho}(t)
=
(\mathbb I\otimes\mathcal L_t)(\rho(t)),
\end{align}
gives
\begin{align}
\frac{d}{dt}
\langle
\mathbb I\otimes O
\rangle_t
&=
{\rm Tr}
\left[
(\mathbb I\otimes\mathcal L_t)(\rho(t))
(\mathbb I\otimes O)
\right].
\end{align}

Introducing the adjoint Liouvillian through
\begin{align}
{\rm Tr}
\left[
\mathcal L_t(\rho)O
\right]
=
{\rm Tr}
\left[
\rho\,
\mathcal L_t^\dagger(O)
\right],
\end{align}
one obtains the Heisenberg equation
\begin{align}
\frac{d}{dt}
\langle
\mathbb I\otimes O
\rangle_t
=
{\rm Tr}
\left[
\rho(t)
(\mathbb I\otimes\mathcal L_t^\dagger(O))
\right].
\end{align}

Throughout this appendix we define
\begin{align}
\mathcal{S}_i(t)
=
\langle
\sigma_i\otimes\sigma_i
\rangle_t,
\qquad
i=x,y,z.
\end{align}


\subsection{Amplitude-damping channel}

The master equation is
\begin{align}
\frac{d\rho}{dt}
=
\mathcal L_t(\rho)
=
\gamma(t)
\left(
\sigma_-\rho\sigma_+
-
\frac12
\{\sigma_+\sigma_-,\rho\}
\right).
\end{align}

The corresponding adjoint generator is
\begin{align}
\mathcal L_t^\dagger(O)
=
\gamma(t)
\left(
\sigma_+O\sigma_-
-
\frac12
\{\sigma_+\sigma_-,O\}
\right).
\end{align}

Setting $O=\sigma_z$ gives
\begin{align}
\mathcal L_t^\dagger(\sigma_z)
=
\gamma(t)
(\mathbb I-\sigma_z).
\end{align}

Therefore,
\begin{align}
\frac{d}{dt}
\mathcal{S}_z(t)
&=
{\rm Tr}
\left[
\rho(t)
\left(
\sigma_z
\otimes
\mathcal L_t^\dagger(\sigma_z)
\right)
\right]
\nonumber\\
&=
\gamma(t)
{\rm Tr}
\left[
\rho(t)
\left(
\sigma_z
\otimes
(\mathbb I-\sigma_z)
\right)
\right]
\nonumber\\
&=
\gamma(t)
\left(
\langle
\sigma_z\otimes\mathbb I
\rangle_t
-
\mathcal{S}_z(t)
\right).
\end{align}

Since the channel $\Lambda_{t,0}$ acts only on the system
qubit and is trace preserving, tracing over the system
eliminates its action, i.e.,
\begin{align}
{\rm Tr}_S
\!\left[
(\mathbb I\otimes\Lambda_{t,0})(X)
\right]
=
{\rm Tr}_S(X)
\end{align}
for any bipartite operator $X$. Therefore,
\begin{align}
\rho_A(t)
&=
{\rm Tr}_S[\rho(t)]
\nonumber\\
&=
{\rm Tr}_S
\left[
(\mathbb I\otimes\Lambda_{t,0})
\left(
|\Phi^+\rangle\langle\Phi^+|
\right)
\right]
\nonumber\\
&=
{\rm Tr}_S
\left[
|\Phi^+\rangle\langle\Phi^+|
\right]
=
\frac{\mathbb I}{2}.
\end{align}
Hence,
\begin{align}
\langle
\sigma_z\otimes\mathbb I
\rangle_t
&=
{\rm Tr}
\left[
\rho_A(t)\sigma_z
\right]
\nonumber\\
&=
\frac12{\rm Tr}(\sigma_z)
=
0,
\end{align}
since $\sigma_z$ is traceless.

Consequently
\begin{align}
\frac{d}{dt}
\mathcal{S}_z(t)
=
-\gamma(t)\mathcal{S}_z(t).
\end{align}


\subsection{Thermal channel}

The master equation is
\begin{align}
\frac{d\rho}{dt}
=
\mathcal L_t(\rho)
=
\gamma_\downarrow(t)
\left(
\sigma_-\rho\sigma_+
-
\frac12
\{\sigma_+\sigma_-,\rho\}
\right)+
\gamma_\uparrow(t)
\left(
\sigma_+\rho\sigma_-
-
\frac12
\{\sigma_-\sigma_+,\rho\}
\right).
\end{align}

The adjoint generator becomes
\begin{align}
\mathcal L_t^\dagger(O)
&=
\gamma_\downarrow(t)
\left(
\sigma_+O\sigma_-
-
\frac12
\{\sigma_+\sigma_-,O\}
\right)+
\gamma_\uparrow(t)
\left(
\sigma_-O\sigma_+
-
\frac12
\{\sigma_-\sigma_+,O\}
\right).
\end{align}

Using
\begin{align}
\sigma_+\sigma_-
=
\frac{\mathbb I-\sigma_z}{2},
\qquad
\sigma_-\sigma_+
=
\frac{\mathbb I+\sigma_z}{2},
\end{align}
one finds
\begin{align}
\mathcal L_t^\dagger(\sigma_z)
=
(\gamma_\downarrow-\gamma_\uparrow)\mathbb I
-
(\gamma_\downarrow+\gamma_\uparrow)\sigma_z.
\end{align}

Therefore,
\begin{align}
\frac{d}{dt}
\mathcal{S}_z(t)
&=
(\gamma_\downarrow-\gamma_\uparrow)
\langle
\sigma_z\otimes\mathbb I
\rangle_t -
(\gamma_\downarrow+\gamma_\uparrow)
\mathcal{S}_z(t).
\end{align}

Again,
\begin{align}
\langle
\sigma_z\otimes\mathbb I
\rangle_t
=
0,
\end{align}
so that
\begin{align}
\frac{d}{dt}
\mathcal{S}_z(t)
=
-
(\gamma_\downarrow+\gamma_\uparrow)
\mathcal{S}_z(t).
\end{align}

For a thermal reservoir satisfying
$\gamma_\downarrow=(N(\Omega_0)+1)\gamma(t)$ and
$\gamma_\uparrow=N(\Omega_0)\gamma(t)$,
this reduces to
\begin{align}
\frac{d}{dt}
\mathcal{S}_z(t)
=
-(2N(\Omega_0)+1)\gamma(t)\mathcal{S}_z(t).
\end{align}


\subsection{Pauli dynamical map}

The master equation is
\begin{align}
\frac{d\rho}{dt}
&=
\mathcal L_t(\rho)
\nonumber\\
&=
\gamma_x(t)
(\sigma_x\rho\sigma_x-\rho)
+
\gamma_y(t)
(\sigma_y\rho\sigma_y-\rho)
+
\gamma_z(t)
(\sigma_z\rho\sigma_z-\rho).
\end{align}

The adjoint Liouvillian is
\begin{align}
\mathcal L_t^\dagger(O)
&=
\gamma_x(t)
(\sigma_xO\sigma_x-O)
+
\gamma_y(t)
(\sigma_yO\sigma_y-O)
+
\gamma_z(t)
(\sigma_zO\sigma_z-O).
\end{align}

Using
\begin{align}
\sigma_i\sigma_j\sigma_i
=
\begin{cases}
\sigma_j,&i=j,\\
-\sigma_j,&i\neq j,
\end{cases}
\end{align}
one obtains
\begin{align}
\mathcal L_t^\dagger(\sigma_x)
&=
-2(\gamma_y+\gamma_z)\sigma_x,
\\
\mathcal L_t^\dagger(\sigma_y)
&=
-2(\gamma_x+\gamma_z)\sigma_y,
\\
\mathcal L_t^\dagger(\sigma_z)
&=
-2(\gamma_x+\gamma_y)\sigma_z.
\end{align}

Consequently,
\begin{align}
\frac{d}{dt}
\mathcal{S}_x(t)
=
-2(\gamma_y+\gamma_z)\mathcal{S}_x(t),
\\
\frac{d}{dt}
\mathcal{S}_y(t)
=
-2(\gamma_x+\gamma_z)\mathcal{S}_y(t),
\\
\frac{d}{dt}
\mathcal{S}_z(t)
=
-2(\gamma_x+\gamma_y)\mathcal{S}_z(t).
\end{align}

Unlike the amplitude-damping channels, no identity operator appears in the adjoint action of the Pauli generator. Therefore these equations hold for arbitrary two-qubit states and do not rely on the Bell-state trajectory.
\end{document}